\documentclass[letterpaper,11pt]{article}
\pdfoutput=1
\usepackage{bm} %bold math
\usepackage[table,xcdraw,svgnames, dvipsnames]{xcolor}
\usepackage{microtype}
\usepackage[utf8]{inputenc}
\usepackage{graphicx,fullpage,paralist}
\usepackage{amsmath, amssymb, amsthm}
\usepackage{comment,hyperref}
\usepackage{thmtools}
\usepackage{tikz}
\usepackage{pgfplots}
\usepackage{varwidth}
\pgfplotsset{compat=1.10}
\usepgfplotslibrary{fillbetween}
\usetikzlibrary{backgrounds}
\usetikzlibrary{patterns}
\usetikzlibrary{positioning}
\usepackage{enumitem}
\usepackage{geometry}
\usepackage{array}
\usepackage{multirow}
\usepackage{stmaryrd}
\usepackage[font=small]{caption}
\usepackage{makecell}
\usepackage{bbold}
\usepackage{cite}
\usepackage{booktabs}
\usepackage[skip=3pt, indent=20pt]{parskip}
\newcommand{\PP}{\mathbb{P}}
\newcommand{\calR}{\mathcal{R}}

\newcommand{\calM}{\mathcal{M}}

\newcommand{\calP}{\mathcal{P}}
\newcommand{\calQ}{\mathcal{Q}}

\newcommand{\calC}{\mathcal{C}}

\newcommand{\calE}{\mathcal{E}}
\newcommand{\calI}{\mathcal{I}}
\newcommand{\calB}{\mathcal{B}}
\newcommand{\calN}{\mathcal{N}}
\newcommand{\calS}{\mathcal{S}}

\newcommand{\calD}{\mathcal{D}}

\newcommand{\calV}{\mathcal{V}}

\newcommand{\bijec}{\mathbf{Bij}}
\newcommand{\NP}{\mathsf{NP}}

\newcommand{\RR}{\mathbb{R}}
\newcommand{\NN}{\mathbb{N}}
\newcommand{\ZZ}{\mathbb{Z}}

\newcommand{\EE}{\mathbb{E}}

\DeclareMathOperator{\OPT}{\mathsf{OPT}}
\DeclareMathOperator{\opt}{\mathsf{OPT}}
\DeclareMathOperator{\alg}{\mathsf{ALG}}
\DeclareMathOperator{\ALG}{\mathsf{ALG}}

\DeclareMathOperator{\Bin}{\mathsf{Bin}}

\newtheorem{theorem}{Theorem}
\newtheorem{lemma}{Lemma}

\newtheorem{definition}{Definition}
\newtheorem{example}{Example}
\newtheorem{proposition}{Proposition}
\newtheorem{claim}{Claim}

\usepackage[textsize=tiny,textwidth=2cm]{todonotes}

\usepackage{multirow}
\usepackage[noend]{algpseudocode}
\usepackage{algorithm,algorithmicx}

\newlength{\algofontsize}
\hypersetup{
	colorlinks,
	linkcolor={red!50!black},
	citecolor={blue!50!black},
	urlcolor={blue!80!black}
}
\begin{document}
\algrenewcommand\algorithmicrequire{\textbf{Input:}}
\algrenewcommand\algorithmicensure{\textbf{Output:}}
	
\title{Online Combinatorial Assignment in Independence Systems}
	
\author{Javier Marinkovic
\thanks{Department of Mathematical Engineering, Universidad de Chile, Chile. {\tt javier.marinkovic95@gmail.com}}
\and Jos\'e~A. Soto			
\thanks{Department of Mathematical Engineering, Universidad de Chile, Chile. {\tt jsoto@dim.uchile.cl}}
\and Victor Verdugo
\thanks{Institute of Engineering Sciences, Universidad de O'Higgins, Chile. {\tt victor.verdugo@uoh.cl}}
}

\date{}

\maketitle
\thispagestyle{empty}
\begin{abstract}

We consider an online multi-weighted generalization of several classic online optimization problems over independence systems, which we call the online combinatorial assignment problem. In this setting, there is an independence system over a ground set of elements and a set of agents that arrive online one by one. Upon arrival, each agent reveals a weight function over the elements of the ground set. The decision-maker has to allocate at most one element to each agent, such that the final allocation is an independent set. Our setting generalizes various scenarios, depending on the online model. For example, when the agents only put positive weight on a single element of the ground set (i.e., single-minded), we recover the classic online matching, matroid secretary, and matroid prophet inequalities, among others. 

One prominent case is when the independence system is given by the matchings of a hypergraph. 
This setting models the fundamental combinatorial auction problem, where every node represents an item to be sold, and every edge represents a bundle of items. For combinatorial auctions, Kesselheim et al. showed upper bounds of $O(\log \log(k)/\log(k))$ and $O(\log \log(n)/\log(n))$ on the competitiveness of any online algorithm, even in the random order model, where $k$ is the maximum size of a bundle and $n$ is the number of items. In our first result, we provide a major exponential improvement on these upper bounds to show that the competitiveness of any online algorithm in the prophet IID setting is upper bounded by $O(\log(k)/k)$, and $O(\log(n)/\sqrt{n})$.

From the algorithmic side, using linear programming, we provide new and improved guarantees for the $k$-bounded online combinatorial auction problem (i.e., bundles of size at most $k$). We show a $(1-e^{-k})/k$-competitive algorithm in the prophet IID model, a $1/(k+1)$-competitive algorithm in the prophet-secretary model using a single sample per agent, and a $k^{-k/(k-1)}$-competitive algorithm in the secretary model.
Furthermore, these three algorithms run in polynomial time. Our algorithms work in more general independence systems where the offline combinatorial assignment problem admits the existence of a certain type of polynomial-time randomized algorithm that we call {\it certificate sampler}. We show that certificate samplers have a nice interplay with the online combinatorial assignment problem in random order models, and we provide new polynomial-time competitive algorithms for several independence systems, including $k$-hypergraph matchings, some classes of matroids, matroid intersections, and matchoids, using greedy and linear programming techniques.
\end{abstract}
\thispagestyle{empty}
\newpage

\section{Introduction}

The secretary and prophet inequality problems are fundamental in optimal stopping theory \cite{Dynkin1963,Hill1982,GM66,SC83,Kertz1986}.
They both model online selection scenarios where the values (weights) of each agent are revealed online, and the decision-maker has to assign an item to only one of them in an irrevocable way.
While in the secretary problem, the weights are arbitrary, and the arrival order of the agents is random, in the prophet inequality scenario, the weights are random, and the arrival order is arbitrary.
Motivated by their numerous applications in online market design and pricing \cite{lucier2017economic}, these problems have been extended to other combinatorial settings, such as matchings and matroids, in order to capture more complicated constraints on the solution constructed, on the fly, by the decision-maker \cite{Ehsani2018,SotoTV2021,Ezra2020OnlineModels,kleinberg2019matroid,Rubinstein2016}.

A central problem in the intersection of economics, optimization, and algorithms is the combinatorial auction problem, where the goal is to assign a bundle of items to each agent in a way that every item is assigned at most once, and the total weight of the solution is maximized, i.e., a maximum welfare allocation.
On top of the intrinsic combinatorial difficulty of this setting, the question becomes even more challenging when the problem is solved online.
This has been a very fruitful research area in recent years, both in the design of approximation algorithms and incentive-compatible mechanisms\cite{assadi2020improved,assadi2021improved,babaioff2014efficiency,baldwin2019understanding,dobzinski2005approximation,dobzinski2016breaking,feldman2013combinatorial,correa2023constant,paes2020computing}.

In this work, we consider a general setting that captures the previous optimization environments.
In the {\it combinatorial assignment problem}, we are given an independence system over a ground set of elements and a set of agents.
Each agent has a non-negative weight function over the elements. 
The independence system typically encodes the existence of constraints over the elements, and an independent set represents a feasible set of elements according to these constraints.
The goal of the decision-maker is to assign at most one element to each agent in a way that the total weight of the solution is maximized and such that the set of assigned elements is an independent set.
For example, when the agents are single-minded (i.e., put all the weight in a single, previously declared element), this setting captures the classic maximum weight matching problem (independence system of matchings in a graph) or the problem of finding the maximum weight base in a matroid.
When the agents are not necessarily single-minded, we can recover the combinatorial auction problem
by taking the independence system of matchings in a hypergraph, where every node represents an item, and every edge is a bundle of items.
The combinatorial assignment problem is, in general, $\NP$-hard, as it captures the $k$-dimensional matching problem \cite{karp2010reducibility}.

In the online combinatorial assignment problem, the agents arrive sequentially, and upon arrival, they present a weight function over the elements.
The decision-maker decides whether to assign an element to the agent, subject to the constraint that the set of elements assigned is an independent set.
The goal is to find an assignment with total weight as close as possible to the optimal weight achievable in the {\it offline} setting, i.e., when the decision-maker has full information.
We consider random-order arrival models: The prophet IID, prophet-secretary, and secretary models.

We also consider the {\it single-sample} data-driven variant in the prophet-secretary model, where the decision maker can access a single sample of the weight function per agent. 
In Section \ref{subsec:online-models}, we describe the different online models in detail.
In any of these models, an algorithm is $\gamma$-competitive if it constructs a solution for the online combinatorial assignment problem with a total (expected) weight of at least a $\gamma$ fraction of the optimal offline value. 

At the core of our work is the $k$-bounded online combinatorial auction problem \cite{KesselheimRTV13,Correa2022CFPW,dutting2020prophet}, where the agents only assign positive weight to bundles of at most $k$ items. Translating this to our combinatorial assignment model, there is a fixed hypergraph with edges of size at most $k$, and each agent has a non-negative weight function over the edges.

\subsection{Our Contributions and Techniques}

For the $k$-bounded online combinatorial auction problem, Kesselheim et al. \cite{KesselheimRTV13} showed that the competitiveness of any online algorithm, even in random-order models, is upper-bounded by $O(\log \log(k)/\log(k))$ and a $O(\log \log(n)/\log(n))$, where $n$ is the number of items.
We provide an exponential improvement on these bounds to show that the competitiveness of any online algorithm in the prophet IID setting is upper bounded by $O(\log(k)/k)$, and $O(\log(n)/\sqrt{n})$. This result is proved in Section \ref{sec:new-ub-k-matching} (Theorem \ref{thm:upper-bounds-k-matching}).

From the algorithmic side, we provide new and improved guarantees for the $k$-bounded online combinatorial auction problem. We show a $(1-e^{-k})/k$-competitive algorithm in the prophet IID model, a $1/(k+1)$-competitive algorithm in the prophet-secretary model using a single sample per agent, and a $k^{-k/(k-1)}$-competitive algorithm in the secretary model;
these three algorithms run in polynomial time and are based on a linear programming relaxation for the maximum weight $k$-hypergraph matching problem (Theorem \ref{thm:dbca}).
Recently, Ezra et al. \cite{Ezra2020} also achieved the $k^{-k/(k-1)}$-competitiveness in the secretary model, but their algorithm runs in exponential time, and they leave as an open question the possibility of attaining this bound in polynomial time.
Our result answers this question in the affirmative.
In particular, we improve on the $1/(ek)$-competitive algorithm by Kesselheim et al. \cite{KesselheimRTV13} for the secretary model by using a similar algorithmic approach but with a different analysis.
Recently, Correa et al. \cite{Correa2022CFPW} got a $1/(k+1)$-competitive algorithm in the prophet model with full distributional knowledge. 
Our $1/(k+1)$-competitiveness in the prophet-secretary model states that we can replace the full distributional knowledge assumption with a single sample per agent at the cost of requiring the agents to arrive in random order.
Furthermore, our $(1-e^{-k})/k$-competitiveness for prophet IID states that we can surpass the $1/(k+1)$ barrier when the agents are identically distributed. 
Our algorithm for the IID setting works even if the distribution over the weight functions is unknown, as we only require access to samples from the distribution.

We remark that our new competitiveness guarantees for the $k$-bounded combinatorial auction problem are all $1/k+o(1/k)$, and our new upper bound is $O(\log(k)/k)$, so our results asymptotically close the gaps (up to a logarithmic factor) in the following models: IID, prophet, prophet secretary, single-sample prophet secretary and secretary; see Table \ref{big-table} for a comparison. The only models in which the gap is larger are the single-sample prophet model (and the harder order-oblivious model), in which the best lower-bound is still $\Omega(1/k^2)$ \cite{KP2019}. This indicates that the single-sample prophet model is much harder than its random-order counterpart and the model with full-distributional knowledge. A similar phenomenon occurs in the case of matroids. For the single-minded version, constant competitive algorithms exist for IID, prophet, and prophet-secretary \cite{kleinberg2019matroid,Ehsani2018}. For the rest of the models (secretary, single-sample prophet, single-sample secretary, and order-oblivious) only a $\Omega(1/\log \log(\text{rank}))$ lower bound is known \cite{FeldmanSZ18}. 

Our algorithms and guarantees work in more general independence systems where the offline combinatorial assignment problem admits a polynomial-time randomized algorithm that we call {\it certificate sampler}.
In Section \ref{sec:algorithms}, we show our framework based on certificate samplers for the prophet IID (Theorem \ref{thm:iid-template}), prophet-secretary (Theorem \ref{thm:sample-template}), and secretary (Theorem \ref{thm:sec-template}) models and prove competitiveness guarantees parameterized on two values, $\gamma$ and $k$. The value $\gamma$ is the offline approximation guarantee of the certificate sampler, and $k$ captures what we call a {\it blocking} property of the certificate sampler: $k$ represents the maximum number of elements that could prevent another element from being part of an independent system.
Our certificate samplers are not required to give a feasible assignment; they only need to provide guarantees on expectation. 
We exploit this in Section \ref{sec:polytime-certifiers} to get polynomial-time guarantees for the three online models via linear programming and greedy algorithms in some matroids (Theorem \ref{thm:matroids}), matroid intersection, and matchoids (Theorem \ref{thm:matroid-intersection}).
Our certificate samplers play a role similar to that of \emph{online contention resolution schemes}, which yield state-of-the-art algorithms for some problems in prophet models \cite{FeldmanSZ2021,Brubach2021ImprovedSchemes,Pollner2022ImprovedEconomy,MacRury2023,Ezra2020OnlineModels,adamczyk2018random}.

\noindent {\bf Summary of new and previous bounds.}  In Section \ref{sec:table}, we present in Table \ref{big-table} a panorama of the lower and upper bounds known for each of the online models, both for the single-minded and the general $k$-bounded combinatorial auction problem, including our new results. 

\section{Preliminaries}\label{sec:prelim}

An independence system $(E,\calI)$ is a pair where $E$ is a finite set and $\calI$ is a nonempty collection of subsets of $E$, called independent sets.
The collection $\calI$ is downward closed; namely, any subset of an independent set is independent. 
In what follows, we formally present the combinatorial assignment problem and the online models considered in this work.

\subsection{Combinatorial Assignment in Independence Systems}

For an independence system $(E,\calI)$ and a set $A$ of $m$ agents, a function $M\colon A\to E\cup \{\bot\}$ is a {\it feasible assignment} of the agents if the following conditions hold:
\begin{enumerate}[label=(\Roman*)]
    \item For every pair of agents $a,a'\in A$ with $a\ne a'$, if $M(a)=M(a')$ then $M(a)=M(a')=\bot$.\label{feasible-a} 
    \item The set $\{M(a):a\in A\}\setminus \{\bot\}$ is an independent set in $\calI$.\label{feasible-b}
\end{enumerate}

We interpret a feasible assignment $M$ as follows. When $M(a)=\bot$, agent $a$ has not been assigned any element $e\in E$.
Otherwise, $M(a)\in E$ is the element assigned to agent $a$, and condition \ref{feasible-a} guarantees that two agents cannot be assigned to the same element in $E$.
Condition \ref{feasible-b} guarantees that the set of elements assigned to the agents is an independent set in $\calI$.
A {\it weight profile} for the agents $A$ is a mapping $\bm{w}$ such that for every $a\in A$, $\bm{w}(a)$ is equal to a weight function $w_a\colon E\cup \{\bot\}\to \RR_+$ with $w_a(\bot)=0$.
Given a weight profile $\bm{w}$ for $A$, we consider the problem of computing the maximum weight feasible assignment, that is,
\begin{equation}
\max\Big\{\textstyle\sum_{a\in A}w_a(M(a))\colon M\text{ satisfies \ref{feasible-a} and \ref{feasible-b}}\Big\}.\label{eq:optimal-assignment}
\end{equation}
We denote by $\opt(\bm{w})$ the optimal value of \eqref{eq:optimal-assignment}.
For example, when $(E,\calI)$ is the independence system of matchings in a graph with edges $E$, problem \eqref{eq:optimal-assignment} corresponds to finding a maximum weight matching on 3-uniform hypergraphs, which is, in general, $\mathsf{NP}$-hard.
When $(E,\calI)$ is a matroid, we can solve  \eqref{eq:optimal-assignment} in polynomial time, using algorithms for matroid intersection (see, e.g., \cite{edmonds1979matroid,lawler1975matroid,bernhard2008combinatorial}).
We say that the agents are single-minded if for every  $a\in A$, there exists a single element $f_a\in E$ such that $w_a(f_a)>0$, and $w_a(e)=0$ for every $e\ne f_a$.
With single-minded agents, \eqref{eq:optimal-assignment} corresponds to finding the maximum weight matching, which can be solved efficiently.

When $(E,\calI)$ is the set of matchings in a hypergraph with edges $E$, problem \eqref{eq:optimal-assignment} corresponds to the combinatorial auction problem \cite{KesselheimRTV13,Correa2022CFPW,dutting2020prophet}.
In this case, every node in the hypergraph represents an item to be sold, and every edge represents a bundle of items.
Each agent has a weight function over bundles, and \eqref{eq:optimal-assignment} corresponds to finding the assignment of bundles to agents that maximize the total weight.
When the maximum edge size (i.e., maximum bundle size) is $k$, the independence system is the $k$-hypergraph matchings, and \eqref{eq:optimal-assignment} is the $k$-bounded combinatorial auction problem.

\subsection{Online Models and Competitiveness}\label{subsec:online-models}

In what follows, we describe the main online models considered in this work.

\noindent{\bf Prophet model.} 
Every agent $a\in A$ has a distribution $\calD_a$ over the weight functions.
The agents then arrive in an arbitrary order $a_1,\ldots,a_m$, and each $a_t$, independently from the rest, upon arrival reveals a random weight function $r_{a_t}\sim \calD_{a_t}$. 
The decision maker decides irrevocably whether to assign agent $a_t$ to some element $e\in E$ or not; if the agent is not assigned to an element in $E$, it is assigned to $\bot$.
Furthermore, the assignment of agents should satisfy \ref{feasible-a}-\ref{feasible-b}.
We are in the \textbf{prophet IID model} when the distributions are all identical.
We say that an algorithm is $\gamma$-competitive if, for every instance, it constructs a feasible assignment with a total weight that is, on expectation, at least a $\gamma$ fraction of the expected optimal value for the problem \eqref{eq:optimal-assignment}, that is, $\EE_{\bm{r}}[\opt(\bm{r})]$.

\noindent{\bf Prophet-secretary model.} 
Every agent $a\in A$ has a distribution $\calD_a$ over the weight functions.
Then, in random order, each agent $a$ reveals a random weight $r_a\sim D_a$, independently from the rest of the agents.
Once an agent reveals the random weight function, the decision-maker decides, irrevocably, whether to assign agent $a$ to some element $e\in E$; if the agent is not assigned to an element in $E$, it is assigned to $\bot$.
Furthermore, the assignment of agents should satisfy \ref{feasible-a}-\ref{feasible-b}.
When the agents' weight distributions are unknown to the decision-maker but we have access to samples from the distribution, we are in the single-sample prophet-secretary model with samples. 
In particular, in the \textbf{single-sample prophet secretary model}, the decision maker has access to only one sample per agent, apart from the random weight function revealed upon the agent's arrival. We say that an algorithm is $\gamma$-competitive if, for every instance, it constructs a feasible assignment with a total weight that is, on expectation, at least a $\gamma$ fraction of the expected optimal value for the problem \eqref{eq:optimal-assignment} when we consider the random weights $w_a=r_a$ for every agent $a$, that is, $\EE_{\bm{r}}[\opt(\bm{r})]$.

\noindent{\bf Secretary model.} Each agent $a$ reveals, in random order, a weight function $w_a$, and we have to decide, irrevocably, whether to assign agent $a$ to some element $e\in E$ or not; in case the agent is not assigned to an element in $E$, is assigned to $\bot$.
Furthermore, the assignment of agents should satisfy \ref{feasible-a}-\ref{feasible-b}.
We say that an algorithm is $\gamma$-competitive if, for every instance, it constructs a feasible assignment with a total weight that is, on expectation, at least a $\gamma$ fraction of the optimal value $\opt(\bm{w})$ for the problem \eqref{eq:optimal-assignment}.\footnote{Tipically, algorithms for secretary problems skip a random subset of the agents to learn information about their weights. A variant of the secretary model is obtained by allowing the algorithm to perform first a subset-sampling phase and then process the rest of the agents in adversarial order. This model is usually known as \textbf{order-oblivious}.}

\section{New Upper Bounds for \texorpdfstring{$k$}{k}-bounded Combinatorial Auctions}\label{sec:new-ub-k-matching}

In the following theorem, we give new upper bounds for the case of $k$-bounded combinatorial auctions in the prophet IID model. 
\begin{theorem}\label{thm:upper-bounds-k-matching}
The competitiveness of any algorithm for the $k$-bounded online combinatorial auction problem in the prophet IID model is upper bounded by 
%$O(\min\{\log(k)/k,\log(m)/m,\log(n)/\sqrt{n}\})$,
$O(\log(k)/k)$, $O(\log(m)/m)$ and $O(\log(n)/\sqrt{n})$,
where $m$ is the number of agents, and $n$ is the number of items. 
\end{theorem}

Our bounds in Theorem \ref{thm:upper-bounds-k-matching} represent an exponential improvement over the previous upper bounds of $O(\log\log(k)/\log(k))$ and $O(\log\log(n)/\log(n))$ given by Kesselheim et al. \cite{KesselheimRTV13}, which is inspired by a former upper bound by Babaioff et al. for the secretary problem in independence systems \cite{BabaioffIKK2018}.

We start by describing the hypergraph $G=(V,E)$ used in our proof of Theorem \ref{thm:upper-bounds-k-matching}.
Fix a positive integer number $m$, which is equal to the number of agents in our instance.
The set of nodes $V$ is constructed as follows. 
Consider an $m\times m$ table $T$, and for each entry $(i,J)$ with $i\ne j$, we have two different nodes corresponding to the cell $T(i,j)$; no nodes are associated with the diagonal cells of $T$. 
We call $V_{T}$ this first set of nodes.
Furthermore, for each $i\in [m]$, we have another different node $x_i$.
We call $V_X$ this set of nodes.
The set of nodes in the hypergraph corresponds to $V=V_{T}\cup V_{X}$.
In total, we have $|V|=|V_T|+|V_X|=2(m^2-m)+m=2m^2-m$ nodes.
The set $E$ of edges is constructed as follows.
Let $E_T$ be the set of all subsets of nodes in $V_T$ that can be formed by selecting an index $i\in [m]$, choosing exactly one node from each nondiagonal cell in column $i$ of $T$, and then exactly one node from each nondiagonal cell in row $i$ of $T$.
The set of edges in the hypergraph corresponds to $E=\{f\cup \{x\}:f\in E_{T}\text{ and }x\in V_X\}$.
That is, we take the union between any set in $E_T$ and one node in $V_X$.
Every $f\in E_T$ satisfies $|f|=2(m-1)$, and therefore every edge in $E$ has size $2(m-1)+1=2m-1$.

For each $i\in [m]$, we define a distribution $\calD_i$ over weight functions as follows. 
Let $R_i$ be the set of $2(m-1)$ nodes in row $i$, and let $C$ be a set of size $m-1$ obtained by choosing exactly one node from every non-diagonal cell in column $i$, uniformly at random (i.e., with probability $1/2$ each). 
We remark that $C$ is a random set that depends on $i$.
Consider the  weight function $w_i:E\to \{0,1\}$ defined as follows: $w_{i}(e)=1$ for every $e\in E$ such that $e=f\cup \{x_i\}$ with $f\in E_T$ and $f\subseteq R_i\cup C$, and it is equal to zero otherwise. Note that the weight function $w_i$ is randomized, as it depends on $C$.
The common distribution $\calD$ for all agents is defined as follows: with probability $1/m$, the weight function of the agent is equal to $w_i\sim \calD_i$.
Namely, every agent $a$ picks (with replacement and independently from the rest) a label $\ell(a)$ uniformly at random in $[m]$, and reveals a weight function $w_{\ell(a)}$ distributed according to $\calD_{\ell(a)}$.
Before proving the theorem, we have the following lemma.

\begin{lemma}\label{claim:log-m}
For the instance given by the hypergraph $G$, there is no online algorithm that achieves an expected weight greater than $\log_2(m+1)$. 
\end{lemma}

To prove the lemma, we show that at every step $j$, and if the algorithm has assigned $i$ edges so far, then for all the following steps we have a probability of at most $2^{-i}$ to increase the size of the matching, and so in expectation we need to wait for at least $2^i$ steps for the matching to increase. Therefore, we need at least $1+2+4+8+\dots + 2^{\ell}$ steps to reach a matching of size $\ell$. After the $m$ steps, we should have a matching of weight at most $O(\log(m))$. 

\begin{proof}[Proof of Lemma \ref{claim:log-m}]
Assume, without any loss, that $\ALG$ is an algorithm that never assigns edges of weight zero.
For every $j\in [m]$, let $a_j$ be the $j$-th agent arriving. We denote by $\alg_j\subseteq E$ the current matching computed by the algorithm in each step $j\in [m]$ and after the arrival of agent $a_j$.
For each $j\in [m]$, let $A_{j}$ be the subset of agents in $\{a_1,\ldots,a_j\}$ for which $\ALG_{j}$ has assigned an edge, and let $Z_j=|A_j|-|A_{j-1}|$, where $A_0=\emptyset$.
Since every non-zero weight is equal to one, the expected weight of the solution computed by the algorithm is equal to the expected size of $A_m$, that is, $\EE[|A_m|]$.
In what follows, we show how to upper bound this value.

Fix $j\in [m]$. First, observe that all the labels of the agents in $A_{j-1}$ are different.
This must hold since, otherwise, if we have two different agents $a,b$ with the same label $t$, any edge with weight one for both agents would contain the node $x_t\in V_X$, and therefore the selected edges would intersect.
Let $C(a_j)$ be the realization of the random set $C$ associated with agent $a_j$, which defines the weight function, and let $\calQ_j$ be the event in which there exists an edge $f\in E_T$ with $f\subseteq C(a_j)\cup R_{\ell(a_j)}$ that is disjoint from those in $\ALG_{j-1}$.
Note that $Z_j=1$ implies that event $\calQ_j$ holds.
In this case, it is necessary that for each agent $a\in A_{j-1}$, the set $C(a_j)$ contains the node from $T(\ell(a),\ell(a_j))$ that is not contained in the edge assigned to $a$, since, otherwise, any edge for $a_j$ with a weight equal to one would intersect the edge assigned to some agent $a\in A_{j-1}$.
This happens exactly with probability $1/2^{|A_{j-1}|}$, and therefore, for every $j\in [m]$, and every $S\subseteq \{a_1,\ldots,a_{j-1}\}$, we have
\begin{align}
\PP[Z_j=1\mid A_{j-1} = S]&= \PP\left[\ALG \text{ assigns one edge to $a_j$}\mid A_{j-1}=S\right]\notag \\
&\leq \PP\left[\calQ_j \mid A_{j-1}=S\right]= 2^{-|S|}.\label{eqn:bound2}
\end{align}
Then, for every $j\in [m]$, and every $S\subseteq \{a_1,\ldots,a_{j-1}\}$, we have
\begin{align}
\EE[2^{Z_j}\mid A_{j-1}=S] &= 2^0\cdot \PP[Z_j=0 \mid A_{j-1}=S)+2^1\cdot \PP[Z_j=1\mid A_{j-1}=S]\notag\\
&= 1-\PP[Z_j=1\mid A_{j-1}=S]+2\cdot \PP[Z_j=1\mid A_{j-1}=S]\notag\\ &=  1+\PP[Z_j=1\mid A_{j-1}=S] \leq 1+2^{-|S|},\label{eqn:lb-ineq}
\end{align}
where the inequality holds by inequality \eqref{eqn:bound2}.
Therefore,
\begin{align}
\EE[2^{|A_j|}] &= \EE[2^{|A_{j-1}|} 2^{Z_j}]\notag\\
&= \sum_{S\subseteq \{a_1,\ldots,a_{j-1}\}}\hspace{-.5cm}\EE[2^{|A_{j-1}|} 2^{Z_j} \mid A_{j-1}=S]\cdot \PP[A_{j-1}=S]\notag\\
&=\sum_{S\subseteq \{a_1,\ldots,a_{j-1}\}}\hspace{-.5cm}2^{|S|}\cdot \EE[2^{Z_j} \mid A_{j-1}=S]\cdot \PP[A_{j-1}=S]\notag\\
&\le \sum_{S\subseteq \{a_1,\ldots,a_{j-1}\}}\hspace{-.5cm}2^{|S|}\cdot \PP[A_{j-1}=S]\cdot (1+2^{-|S|})\notag\\
&= \sum_{S\subseteq \{a_1,\ldots,a_{j-1}\}}\hspace{-.5cm}2^{|S|}\cdot \PP[A_{j-1}=S]+\hspace{-.5cm}\sum_{S\subseteq \{a_1,\ldots,a_{j-1}\}}\hspace{-.5cm}\PP[A_{j-1}=S]=\EE[2^{|A_{j-1}|}]+1,
\end{align}
where the first equality holds since $Z_j=|A_{j}|-|A_{j-1}|$, the second and third hold by conditioning on $A_{j-1}$, the inequality holds by \eqref{eqn:lb-ineq}, and the last equalities hold by recovering the expected value of $2^{|A_{j-1}|}$.
By induction, we get $\EE[2^{|A_j|}]\leq \EE[2^{|A_0|}]+j=1+j$ for every $j\in [m]$, and since the function $2^x$ is convex, Jensen's inequality implies that $2^{\EE[|A_j|]}\leq \EE[2^{|A_j|}] \leq j+1$.
Therefore, $\EE[|A_j|]\leq \log_2(j+1)$ for every $j\in [m]$.
We conclude that the expected weight of the matching computed by the algorithm is at most $\EE[|A_m|]\leq \log_2(m+1)$.
\end{proof}

\begin{proof}[Proof of Theorem \ref{thm:upper-bounds-k-matching}]
To prove the upper bound, we claim that the expected weight of an optimal matching in $G$ is $\Omega(m)$.
This bound, together with Lemma \ref{claim:log-m}, implies that the competitiveness of any online algorithm is upper bounded by $O(\log(m)/m)$.
Since the maximum size $k$ of an edge in $G$ is equal to $2m-1$, this also implies the $O(\log(k)/k)$ upper bound.
Finally, the number of nodes is $n=2m^2-m$, that is, $m=\Theta(\sqrt{n})$, and therefore we recover the $O(\log(n)/\sqrt{n})$ upper bound.

In what follows we prove the claim.
First, recall that the non-zero values of the weight function are all equal to one, and therefore we just have to lower bound the expected size of the largest matching.
Consider a fixed realization of the labels $\ell(a)$ for each agent $a\in A$, and let $L$ be the number of different labels obtained.
For each $a$, let $C(a)$ be the realization of the random set $C$ associated with $a$ that defines the weight function $w_{\ell(a)}$.

If $a$ and $b$ are two agents selecting the same label $i$, then all non-zero weighted edges of $a$ and $b$, respectively, contain the item $x_i\in V_X$, i.e., they all intersect.
Therefore, in an optimal matching, we can only get one unit of weight in total from all agents with the same label. 
This implies that the size of the largest matching is upper-bounded by $L$.
On the other hand, let $b_{1},\ldots,b_{L}$ be any subset of $L$ agents such that they all have different labels.
We denote by $\ell_i$ the label $\ell(b_i)$ of agent $b_i$.
We will construct a matching with a size equal to $L$.
To this end, for each agent $b_i$, we construct an edge $e_i$ as follows.
For every pair of agents $b_i\neq b_j$, let $s(i,j)$ be the node in cell $T(\ell_i,\ell_j)$ that is not in $C(b_j)$. 
For every agent $b_i$, and every index $t\in [m]\setminus \{\ell_1,\ldots,\ell_L\}$, let $v(i,t)$ be any node in the cell $T(\ell_i,t)$, and let 
\[
e_i=C(b_i)\cup \{s(i,j) \colon j\in [L]\text{ and }j \neq i\}\cup \{v(i,t):t\in [m]\setminus \{\ell_1,\ldots,\ell_L\}\} \cup \{x_{\ell_i}\}.
\]
We have that $e_i\in E$ by construction, and $w_{\ell_i}(e_i)=1$. 
Furthermore, we note that for all $i\neq j$, the edges $e_i$ and $e_j$ are disjoint.
Indeed, by construction,
$$e_i\cap e_j \subseteq (R_{\ell_i}\cup C(b_i))\cap (R_{\ell_j}\cup C(b_j))\subseteq \text{Nodes}(\ell_i,\ell_j)\cup \text{Nodes}(\ell_j,\ell_i),$$
where $\text{Nodes}(\ell_i,\ell_j)$ are the two nodes associated with the cell $T(\ell_i,\ell_j)$ in the table $T$, and same for $\text{Nodes}(\ell_j,\ell_i)$.
However, $e_i\cap \text{Nodes}(\ell_i,\ell_j)=\{s(i,j)\}$, but $s(i,j)$ is chosen exactly as the node in $\text{Nodes}(\ell_i,\ell_j)$ that is not in $C(b_j)$, and then $s(i,j)$ is not in $e_j$.
Therefore, $e_i\cap \text{Nodes}(\ell_i,\ell_j)$ is disjoint from $e_j\cap \text{Nodes}(\ell_i,\ell_j)$. Symmetrically, $e_j\cap \text{Nodes}(\ell_j,\ell_i)$ is disjoint from $e_i\cap \text{Nodes}(\ell_j,\ell_i)$, and we conclude that the edges $e_i$ and $e_j$ are disjoint.

Then, for every realization of the random sets $C(a)$, we find that the largest matching has a size equal to $L$.
Since every label is sampled uniformly at random in $[m]$, and independently from the rest, the probability that any index $j\in [m]$ is not sampled in the $m$ trials is equal to $(1-1/m)^m$, and therefore $\EE[L]\ge m(1-(1-1/m))^m\ge m(1-1/e)$.
We conclude that the expected weight of the optimal matching (i.e., the expected size of the largest matching) is at least $m(1-1/e)=\Omega(m)$.
This finishes the proof of the theorem.     
\end{proof}

\section{Algorithmic Framework for General Independence Systems}\label{sec:algorithms}

In this section, we develop a general framework to design algorithms for the combinatorial assignment problem in independence systems, with random-order online models.
To this end, we first introduce a combinatorial object in independence systems that interacts nicely with our online algorithms.
\begin{definition}\label{def:certifier}
A certifier for an independence system $(E,\calI)$ is a tuple $(\calS,\calN,\calB)$ where $\cal{S}$ is a finite set, and $(\calN,\calB)$ is a digraph satisfying the following properties:
\begin{enumerate}[itemsep=0pt,label=(\alph*)]
    \item  $\mathcal{N}\subseteq \calS\times E$ and $\mathcal{B}\subseteq \cal{N}\times\cal{N}$. \label{cert-a}
    \item For all $e\in E$, and $S_1, S_2 \in \calS$ such that $(S_1,e),(S_2,e)\in \mathcal{N}$, we have $((S_1,e),(S_2,e))\in \mathcal{B}$.\label{cert-b}
    \item For every sequence of nodes $(S_1,e_1), (S_2,e_2), \ldots ,(S_t,e_t) \in \calN$ such that $((S_i,e_i),(S_j,e_j))\not\in \cal{B}$ for all $1\leq i<j\leq t$, we have $\{e_1,\dots, e_t\}\in \cal{I}$.\label{cert-c}
\end{enumerate}
Every node in $\calN$ is called a {\it certificate}, 
and if $((S,e),(S',e'))\in \mathcal{B}$ we say that $(S,e)$ blocks the $(S',e')$.
A sequence of certificates $(S_1,e_1),(S_2,e_2),\ldots,(S_t,e_t)$ is called a certification if $(S_i,e_i)$ does not block $(S_j,e_j)$ for every $1\leq i<j\leq t$.
\end{definition}
In words, from \ref{cert-a}, we find that in the digraph $(\calN,\calB)$ each node (certificate) consists of $S\in \cal{S}$ and an element $e\in E$ of the ground set.
The property \ref{cert-b} implies that the digraph has a loop at every node $(S,e)$, since we can take $S_1=S_2=S$, and that certificates for the same element $e\in E$ block each other.
The property \ref{cert-c} also implies that only non-loops of $(E,\calI)$ can have certificates (certificates are certifications of length one).
The following simple proposition summarizes a property that has been repeatedly used in the analysis of our algorithms.

\begin{proposition}\label{prop:certifications}
For every certification $(S_1,e_1),(S_2,e_2),\ldots,(S_t,e_t)$, we have that $\{e_1,\ldots,e_t\}$ is an independent set of size $t$.
\end{proposition}

\begin{proof}
Property \ref{cert-c} states that for every certification $(S_1,e_1),(S_2,e_2),\ldots,(S_t,e_t)$ the corresponding set of items $Q=\{e_1,\dots, e_t\}$ is an independent set. Furthermore, due to property \ref{cert-a}, all the elements in $Q$ must be different (otherwise the corresponding certificates block each other) so the set $Q$ has exactly $t$ elements. 
\end{proof}
For every independence system $(E,\calI)$, we can construct its {\it canonical certifier} as follows: Consider $\calS=\calI$, and the digraph $(\calN,\calB)$ defined as follows: $\mathcal{N}=\{(I,e)\colon I\in \mathcal{I}, e\in I\}$, and $\calB=(\mathcal{N}\times \mathcal{N})\setminus \{ ((I,e),(I,f))\colon e\neq f\}$.
Namely, the certificates for any element $e$ are those pairs $(I,e)$ where $I$ is an independence set containing $e$.  
In the following example, we construct a different certifier for $k$-hypergraph matchings.

\begin{example}\label{example:hypergraph}
Consider the independence system of $k$-hypergraph matchings in $G=(V,E)$. 
Let $\calS=E$, and let $(\mathcal{N},\calB)$ be the following digraph: $\mathcal{N}=\{(e,e)\colon e\in E\}$, and $\calB=\{((e,e),(f,f))\colon e\cap f\ne \emptyset\}$.
Namely, every edge defines its own certificate, and we have that $(e,e)$ blocks $(f,f)$ if $e$ and $f$ have at least one node in common. Properties \ref{cert-a}-\ref{cert-b} are satisfied by construction, and each certification is of the form $(e_1,e_1),(e_2,e_2),\dots, (e_t,e_t)$ where for all $1\le i< j\le t$, the edges $e_i$ and $e_j$ are disjoint, and therefore $\{e_1,\dots, e_t\}$ is a matching.
\end{example}
\begin{definition}\label{def:certificate-sampler}
Given a certifier $(\calS,\calN,\calB)$ for an independence system $(E,\calI)$, a certificate sampler is a randomized algorithm $\mathcal{P}$ that for any set of agents $A$, receives a weight profile $\bm{w}$ for $A$, and outputs a certificate $\calP_a(\bm{w})\in \calN\cup\{(\bot,\bot)\}$ for every $a\in A$. 
We denote by $\calP(\bm{w})$ the output of $\calP$ in input $\bm{w}$, and for every $a$, we denote $S_{a}(\calP,\bm{w})=S$ and $e_a(\calP,\bm{w})=e$ when $\calP_a(\bm{w})=(S,e)$.
\end{definition}

We remark that we allow the certificate sampler to output $(\bot,\bot)\not\in \mathcal{N}$ for zero or more agents in $A$. 
We assume that $(\bot,\bot)$ does not block, nor is blocked, by any certificate in $\calN$, i.e., $((\bot,\bot),(S,e))\not\in \calB$ and $((S,e),(\bot,\bot))\not\in \calB$.
We do not require that $\{e_a(\calP,\bm{w}):a\in A\}\setminus \{\bot\}$ is an independent set,
and we do not require that different $a,a'\in A$ receive different elements $e_a(\calP,\bm{w})$ and $e_{a'}(\calP,\bm{w})$. 

\begin{definition}\label{def:good-sampler}
Consider an independence system $(E,\calI)$,
and a certificate sampler $\calP$ with certifier $(\calS,\calN,\calB)$.
We say that $\calP$ is a $(\gamma,k)$-certificate sampler, with $\gamma\ge 0$ and $k\in \NN$, if for any set $A$, and any weight profile $\bm{w}$ for $A$, the following holds:
\begin{enumerate}[label=(\roman*)]
    \item {\bf (Approximation)} $\EE[\sum_{a\in A}w_{a}(e_a(\calP,\bm{w}))] \geq \gamma \opt(\bm{w})$. \label{sampler-b}
    \item {\bf (Blocking)} For every $(S,e)\in \calN$, we have $\sum_{a\in A}\PP[(\calP_a(\bm{w}),(S,e))\in \calB]\le k$. \label{sampler-c}
\end{enumerate}
\end{definition}
In properties \ref{sampler-b} and \ref{sampler-c}, the probability and expectation are taken only over the internal randomness of the algorithm $\calP$. 
Recall that here we use the fact that $w_a(\bot)=0$ for all weight functions in order to evaluate the summation.

\subsection{Prophet IID Model}\label{sec:iid}

In Algorithm \ref{alg:iid-template}, we provide a template that, given a certificate sampler $\calP$ with certifier $(\calS,\calN,\calB)$, constructs a solution satisfying \ref{feasible-a}-\ref{feasible-b}.
The agents are presented in an arbitrary fixed order $a_1,\ldots,a_m$, meaning that $a_1$ is the first agent being presented, $a_2$ is the second agent being presented, and so on.
Once presented, the agent $a_t$ reveals a random weight function $r_t$. 
The algorithm then selects an index $\ell_t$ uniformly at random in $[m]$ and creates a new weight profile $\calR_t$,  where $\calR_t(j)=R_{t,j}\sim \calD$ for every $j\ne \ell_t$, and $R_{t,\ell_t}=r_t$. Then we run the certificate sampler $\calP$ on the weight profile $\calR_t$, and we assign the element $e_{\ell_t}(\calP,\calR_t)$ to agent $a_t$ as long as it is not blocked by any other previously assigned element.
\begin{algorithm}[H]
    \begin{algorithmic}[1]
    \Require{An instance with $n$ agents $A$.}
    \Ensure{A feasible assignment of the agents.}
    \State Initialize $\alg(a_t)=\bot$ for every $t\in [n]$.   
    \State {\bf Selection:} Initialize an empty sequence $C$ of certificates in $\calN$.
    \For{$t=1$ to $m$}
        \State Run $\mathcal{P}$ on the weight profile $\calR_t$.
        \If{$\calP_{\ell_t}(\calR_t)\neq (\bot,\bot)$ and every certificate in $C$ does not block $\calP_{\ell_t}(\calR_t)$} 
            \State Update $\alg(a_t)=e_{\ell_t}(\calP,\calR_t)$, and append $\calP_{\ell_t}(\calR_t)$ to the end of $C$.\label{line:iid-certification}
        \EndIf    
    \EndFor
    \State Return $\alg$.
    \end{algorithmic}
    \caption{Template for the prophet IID model}
    \label{alg:iid-template}
\end{algorithm}
Recall that $\bm{r}$ is the weight profile that maps every agent $a$ to the weight function $r_a$. Denote as $\mathcal{D}^A$ to the distribution of profile $\bm{r}$. By construction, the profiles  $\mathcal{R}_1,\dots, \mathcal{R}_m$ are independent and identically distributed according to $\mathcal{D}^A$. The following is the main result of this section.

\begin{theorem}\label{thm:iid-template} 
Let $(E,\calI)$ be an independence system, and let $\calP$ be a $(\gamma,k)$-certificate sampler.
Then, Algorithm \ref{alg:iid-template} is $\gamma(1-e^{-k})/k$-competitive for the prophet IID model.
\end{theorem}

Algorithm \ref{alg:iid-template} only requires $O(m^2)$ samples from the unknown distribution $\calD$. Therefore, we recover the same $\gamma(1-e^{-k})/k$ factor for the more restrictive prophet IID model with samples where the distribution $\calD$ of the weights is unknown.
In Lemma \ref{lem:iid-template-correct}, we show that Algorithm \ref{alg:iid-template} computes a feasible assignment for the problem \eqref{eq:optimal-assignment}. 

\begin{lemma}\label{lem:iid-template-correct}
The solution $\alg$ computed by Algorithm \ref{alg:iid-template} satisfies properties \ref{feasible-a} and \ref{feasible-b}.
\end{lemma}
\begin{proof}
We first show by induction that the sequence $C$ maintained by Algorithm \ref{alg:iid-template} is a certification for every iteration $t\in \{1,\ldots,m\}$.
The base case follows since certificates are certifications of length one.
We denote by $(X_{q},f_{q})$ be the $q$-th certificate that is appended in line \ref{line:iid-certification} of Algorithm \ref{alg:iid-template}. 
Suppose that for a certain agent $a_t$ we append the $(q+1)$-th certificate to $C$, that is, $(X_{q+1},f_{q+1})=\calP_{a_t}(\calR_t)$.
By the inductive step, $(X_i,f_i)$ does not block $(X_j,f_j)$ for every $1\le i<j\le q$, and by the condition on line \ref{line:iid-certification} we have that $(X_i,f_i)$ does not block $(X_{q+1},f_{q+1})=\calP_{a_t}(\calR_t)$ for every $1\le i\le q$, and therefore the sequence $(X_1,f_1),\ldots,(X_{q+1},f_{q+1})$ is a certification. 

The image of $\alg$ is exactly equal to the set $\{f_1,\ldots,f_{|C|}\}\cup \{\bot\}$, where $|C|$ is the length of the sequence $C$ at the end of the algorithm execution, and by Proposition \ref{prop:certifications} we have that $\{f_1,\ldots,f_{|C|}\}$ is an independent set, and all its elements are different.
Therefore, $\alg$ satisfies \ref{feasible-a} and \ref{feasible-b}.
\end{proof}

In the following lemma, we use the properties of the certificate sampler $\calP$ to lower bound the expected contribution of any agent $a_t$, for which the certificate sampler gives a certificate $(S,e)$ in the algorithm. Intuitively, the lemma states that the probability that the assignment made by the algorithm on step $i$ is compatible with (it does not block) the selection made on step $t$ is at $(1-k/m)$. In the analysis, we use subindices under $\EE$ or $\PP$ to denote the random sources over which we are taking expectations of probability: subindex $t$ is used to represent all the randomness occurring on the $t$-th iteration, that is, the choice of $\calR_t$ (which includes $\bm{r}_{a_t}$), the choice of $\ell_t$,  and all the internal choices of the certificate sampler on that iteration. In this model, events depending on different subindices are mutually independent.

\begin{lemma}\label{lem:iid-template-prod}
Suppose that in Algorithm \ref{alg:iid-template}, $\calP$ is a $(\gamma,k)$-certificate sampler.
Then, for every value $t\in \{2,\ldots,m\}$, and every $(S,e)\in \calN\cup \{(\bot,\bot)\}$, we have
\begin{align*}
&\EE_{1,2,\dots, t}\left[r_t(e_{\ell_t}(\calP,\calR_t)) \cdot \mathbb{1}\left[\calP_{\ell_t}(\calR_t)=(S,e)\right] \cdot \mathbb{1}\left[\bigwedge_{i=1}^{t-1} \calP_{\ell_i}(\calR_i) \text{ does not block } (S,e) \right]\right]\\
&\ge \EE_t\bigl[r_t(e_{\ell_t}(\calP,\calR_t)) \cdot \mathbb{1}\left[\calP_{\ell_t}(\calR_t)=(S,e)\right]\bigr]\cdot  \left(1-\frac{k}{m}\right)^{t-1}.
\end{align*}
\end{lemma}

\begin{proof}
Since $\calP_{\ell_i}(\calR_i)$ depends only on the randomness of the $i$-th iteration we conclude that the left hand side of the expression to prove equals
\begin{align}
&\EE_t\bigl[r_t(e_{\ell_t}(\calP,\calR_t)) \cdot \mathbb{1}\left[\calP_{\ell_t}(\calR_t)=(S,e)\right]\bigr] \cdot  \prod_{i=1}^{t-1} \EE_{i}\left[\mathbb{1}\left[\calP_{\ell_i}(\calR_i) \text{ does not block } (S,e)\right] \right]\notag\\
&=\EE_t\bigl[r_t(e_{\ell_t}(\calP,\calR_t)) \cdot \mathbb{1}\left[\calP_{\ell_t}(\calR_t)=(S,e)\right]\bigr] \cdot  \prod_{i=1}^{t-1} \PP_{i}\left[\calP_{\ell_i}(\calR_i) \text{ does not block } (S,e)\right].\label{iid:partitioning2}
\end{align}
Note that if $\ell$ is a uniformly chosen agent, then for every $i\in \{1,\ldots,t-1\}$, the random variable $\calP_{\ell_i}(\calR_i)$ has the same distribution 
as $\calP_{\ell}(\bm{r})$, since both $\calR_i$ and $\bm{r}$ have the same distribution $\calD^A$, and $\ell_i$ is a uniform random agent. In particular, $\PP_i\left[\calP_{\ell_i}(\calR_i) \text{ blocks } (S,e) \right] = \EE_{\bm{r}}\frac{1}{m}\sum_{\ell=1}^m\Pr[\calP_\ell (\bm{r}) \text{ blocks} (S,e)] \leq k/m$, where the inequality holds by the blocking property \ref{sampler-c} of the certificate sampler $\calP$. Then, overall, we get that \eqref{iid:partitioning2} can be lower bounded by $\EE_t\bigl[r_t(e_{\ell_t}(\calP,\calR_t)) \cdot \mathbb{1}\left[\calP_{\ell_t}(\calR_t)=(S,e)\right]\bigr]  \cdot \prod_{i=1}^t(1-k/m)$, which finishes the proof of the lemma. 
\end{proof}

\begin{proof}[Proof of Theorem \ref{thm:iid-template}]
Observe that for every $t$, by fixing the weight profile $R_t$ first, the expected contribution of agent $a_t$ to the solution returned by the algorithm is $\EE_{1,\dots, t}[r_t(\ALG(a_t))]]$. Therefore, the expected total value of the solution can be computed as
\begin{align}
&\sum_{t=1}^m \EE_{1,2,\dots, t}\left[r_{t}(e_{\ell_t}(\calP,\calR_t))\cdot \mathbb{1}\left[\bigwedge_{i=1}^{t-1} \calP_{\ell_i}(\calR_i) \text{ does not block } \calP_{\ell_t}(\calR_t) \right]\right] \notag\\
&=\sum_{t=1}^m \sum_{(S,e)\in \mathcal{N}\cup \{(\bot,\bot)\}}\hspace{-0.7cm}\EE_{1,2,\dots, t}\left[r_t(e_{\ell_t}(\calP,\calR_t)) \cdot \mathbb{1}\left[\calP_{\ell_t}(\calR_t)=(S,e)\right] \cdot \mathbb{1}\left[\bigwedge_{i=1}^{t-1} \calP_{\ell_i}(\calR_i) \text{ does not block } (S,e) \right]\right]\notag\\
&\ge \sum_{t=1}^m \sum_{(S,e)\in \mathcal{N}\cup \{(\bot,\bot)\}}\hspace{-0.7cm} \EE_t\bigl[r_t(e_{\ell_t}(\calP,\calR_t)) \cdot \mathbb{1}\left[\calP_{\ell_t}(\calR_t)=(S,e)\right]\bigr]\cdot  \left(1-\frac{k}{m}\right)^{t-1}\notag \\
&=\sum_{t=1}^m\EE_{t}\left[r_{t}(e_{\ell_t}(\calP,\calR_t))\right]\cdot\left(1-\frac{k}{m}\right)^{t-1},\label{iid:transition1}
\end{align}
where the first equality holds by partitioning on the realization of $\calP_{\ell_t}(R_t)$, the inequality holds by Lemma \ref{lem:iid-template-prod}, and the last equality is obtained by recovering the expectation of $r_{t}(e_{\ell_t}(\calP,\calR_t))$.
Observe that the random variable $\left[r_{t}(e_{\ell_t}(\calP,\calR_t))\right]= \left[R_{t,\ell_t}(e_{\ell_t}(\calP,\calR_t))\right]$, can be obtained by first taking a random profile $\bm{r}\sim \calD^A$ and then choosing an index $\ell\in [m]$ uniformly at random and returning $w_\ell(e_{\ell}(\calP,\bm{w})$. Then, $\EE_t\left[r_{t}(e_{\ell_t}(\calP,\calR_t))\right] = \EE_{\bm{w}} \left[\frac{1}{m}\sum_{\ell=1}^m w_\ell(e_{\ell}(\calP,\bm{w})\right]\geq \frac{\gamma}{m} \EE_{\bm{w}}[\OPT(\bm{w})]$, where the inequality is is a consequence of the approximation property \ref{sampler-b} of the certificate sampler $\calP$ when applied in the weight profile $\bm{w}$.

Using that $\bm{r}$ and $\bm{w}$ have the same distribution, then, from \eqref{iid:transition1}, the expected weight of the solution $\alg$ can be lower bounded as follows:
\begin{align}
\EE_{\bm{r}}[\OPT(\bm{r})]
\frac{\gamma}{m}\sum_{t=1}^m\left(1-\frac{k}{m}\right)^{t-1} &= \EE_{\bm{r}}[\OPT(\bm{r})] \frac{\gamma}{k}\left(1-\left(1-\frac{k}{m}\right)^m\right) \ge \EE_{\bm{r}}[\opt(\bm{r})]\frac{\gamma}{k}\left(1-e^{-k}\right) ,\notag
\end{align}
where the equality holds by solving the geometric summation, and the inequality holds since $(1-k/m)^{m}\le \exp(-k)$ for every $m$.
This concludes the proof of the theorem.
\end{proof}

\subsection{Prophet-secretary Model}\label{sec:single-sample}

In Algorithm \ref{alg:sample-template}, we provide a template that, given a certificate sampler $\calP$ with certifier $(\calS,\calN,\calB)$, constructs a solution satisfying \ref{feasible-a}-\ref{feasible-b}.
The algorithm can access a single sample $s_a\sim \calD_a$ of the random weight function for each agent $a\in A$. 
Then, the algorithm sees the agents in random order $\mu\in \bijec([m],A)$, which means that $\mu(1)$ is the first agent that reveals a random weight function $r_{\mu(1)}$, $\mu(2)$ is the second agent that reveals a random weight function $r_{\mu(2)}$, and so on.
For a set of agents $B\subseteq A$, we denote by $\bm{r}(B)$ the weight profile $\{r_b:b\in B\}$, and by $\bm{s}(B)$ the weight profile $\{s_b:b\in B\}$.
For each $t\in [m]$, we denote by  $\calV(\mu,t)$ the weight profile $\bm{r}(\mu_t)\cup \bm{s}(A\setminus\mu_t)$, where $\mu_t=\{\mu(1),\ldots,\mu(t)\}$.
\begin{algorithm}[H]
    \begin{algorithmic}[1]
    \Require{An instance with $m$ agents $A$, revealed in random order $\mu$.}
    \Ensure{A feasible assignment of the agents.}
    \State Initialize $\alg(a)=\bot$ for every agent $a\in A$.   
    \State {\bf Sample:} For every $a\in A$ we have access to a sample $s_a\sim D_a$. 
    \State {\bf Selection:} Initialize an empty sequence $C$ of certificates in $\calN$.
    \For{$t=1$ to $m$}
        \State The agent $\mu(t)$ reveals a weight function $r_{\mu(t)}\sim D_{\mu(t)}$. 
        \State Run $\mathcal{P}$ on the weight profile $\calV(\mu,t)$.
        \If{$\calP_{\mu(t)}(\calV(\mu,t))\neq (\bot,\bot)$ and every certificate in $C$ does not block $\calP_{\mu(t)}(\calV(\mu,t))$} 
            \State Update $\alg(\mu(t))=e_{\mu(t)}(\calP,\calV(\mu,t))$, and append $\calP_{\mu(t)}(\calV(\mu,t))$ to the end of $C$.\label{line:sample-certification}
        \EndIf    
    \EndFor
    \State Return $\alg$.
    \end{algorithmic}
    \caption{Template for the prophet-secretary model}
    \label{alg:sample-template}
\end{algorithm}

The following is the main result of this section.
\begin{theorem}\label{thm:sample-template} 
Let $(E,\calI)$ be an independence system, and let $\calP$ be a $(\gamma,k)$-certificate sampler.
Then, Algorithm \ref{alg:sample-template} is $\gamma/(k+1)$-competitive for the single-sample prophet-secretary model.
\end{theorem}

In particular, Theorem \ref{thm:sample-template} also implies the same guarantee for the prophet-secretary model, where we have access to the distributions of the agents.

In our proofs we will need to deal with three sources of randomness: the random weight functions (summarized as profiles) $\bm{s},\bm{r}$, the random order $\mu$ and the internal randomness of the certificate sampler $\calP$. 
Here, when we use $\EE$ or $\PP$ without a subindex, we are computing an event's expectation or probability on the probability space induced only by the internal randomness of the certificate sampler.
Before proving the theorem, in the following lemmas, we use the properties of the certificate sampler $\calP$ to lower bound the expected contribution of the agent arriving at time $t$, for which the certificate sampler gives a certificate $(S,e)$ in the algorithm. Intuitively, the lemma states that the contribution of the agent is some expected value, which does not depend on $\bm{s}$ anymore, times the probability that no choice made by the algorithm previously blocks $(S,e)$. This probability is not independent for each time step. However, we can lower bound it by a product in which the term associated to step $i$ is $(1-k/(m-t+i))_+$, and can be interpreted as the probability of succeeding at that step, given \emph{any fixed choice of agents arriving between step $i$ and step $t$}.

\begin{lemma}\label{lem:sample-template-prod}
Suppose that in Algorithm \ref{alg:sample-template}, $\calP$ is a $(\gamma,k)$-certificate sampler.
Then, for every value $t\in \{1,\ldots,m\}$, and every $(S,e)\in \calN\cup \{(\bot,\bot)\}$ we have
\begin{align}
&\EE_{\bm{s},\bm{r}}\left[\EE_{\mu}\left[r_{\mu(t)}(e)\cdot \PP\left[\calP_{\mu(t)}(\calV(\mu,t))=(S,e)\wedge \bigwedge_{i=1}^{t-1} \calP_{\mu(i)}(\calV(\mu,i)) \text{ does not block } (S,e) \right]\right]\right]\notag \\
&\ge \EE_{\bm{r}}\left[\EE_{\mu}\left[r_{\mu(t)}(e)\cdot\PP[\calP_{\mu(t)}(\bm{r})=(S,e)]\right]\right] \cdot \prod_{i=1}^{t-1}\left(1-\frac{k}{m-t+i}\right)_+,\label{eq:lemmasample}
\end{align}
where, if $t=1$ the empty product is interpreted as 1. 
\end{lemma}

\begin{proof}
Throughout this proof, for any agent $a$ and any weight profile $\bm{w}$, let $\calE_{a}(\bm{w})$ denote the event in which $\calP_{a}(\bm{w})$ does not block $(S,e)$. By expanding the expectation over $\mu$, the left hand side of \eqref{eq:lemmasample} equals
\begin{align}
&\frac{1}{m!}\hspace{-3pt}\sum_{\mu \in \bijec([m],A)}\hspace{-7pt}\EE_{\bm{s},\bm{r}}\left[r_{\mu(t)}(e)\cdot \PP\left[\calP_{\mu(t)}(\bm{r}(\mu_t)\cup \bm{s}(A\setminus \mu_t))=(S,e)\wedge \bigwedge_{i=1}^{t-1} \calE_{\mu(i)}(\bm{r}(\mu_i)\cup \bm{s}(A\setminus \mu_i)) \right]\right]\notag\\
&=\frac{1}{m!}\hspace{-3pt}\sum_{\mu \in \bijec([m],A)}\hspace{-7pt}\EE_{\bm{s},\bm{r}}\left[r_{\mu(t)}(e)\cdot\PP\left[\calP_{\mu(t)}(\bm{r}(\mu_t)\cup \bm{s}(A\setminus \mu_t))=(S,e)\right]\cdot \PP\left[\bigwedge_{i=1}^{t-1}\calE_{\mu(i)}(\bm{r}(\mu_i)\cup \bm{s}(A\setminus \mu_i))\right]\right],\label{sample:wedge-prob}
\end{align}
where we use that for every fixed order $\mu$ fixed $\bm{s}$ and $\bm{r}$, we have that the random variable $\calP_{\mu(t)}(\bm{r}(\mu_t)\cup \bm{s}(A\setminus \mu_t))$ and all the events $\calE_{\mu(i)}(\bm{r}(\mu_i)\cup \bm{s}(A\setminus \mu_i))$ are stochastically independent (each of them only depends on the internal randomness of a different run of $\calP$). 

Now we fix any order $\mu$. This defines exactly one term in the inner summation of \eqref{sample:wedge-prob}. Observe that among all the random weight functions in the weight profiles $\bm{r}$ and $\bm{s}$, the term within the expectation does not depend on the weight functions of $\bm{r}(A\setminus \mu_t)=\{r_{\mu(t+1)},\dots, r_{\mu(m)}\}$. Therefore, 
\begin{align}
&\EE_{\bm{s},\bm{r}}\left[r_{\mu(t)}(e)\cdot\PP\left[\calP_{\mu(t)}(\bm{r}(\mu_t)\cup \bm{s}(A\setminus \mu_t))=(S,e)\right]\cdot \PP\left[\bigwedge_{i=1}^{t-1} \calE_{\mu(i)}(\bm{r}(\mu_i)\cup \bm{s}(A\setminus \mu_i))\right]\right]\notag\\
&=\EE_{\bm{r}(\mu_t)}\EE_{\bm{s}(\mu_t)}\EE_{\bm{s}(A\setminus \mu_t)}\left[r_{\mu(t)}(e)\cdot\PP\left[\calP_{\mu(t)}(\bm{r}(\mu_t)\cup \bm{s}(A\setminus \mu_t))=(S,e)\right]\prod_{i=1}^{t-1}\PP[ \calE_{\mu(i)}(\bm{r}(\mu_i)\cup \bm{s}(A\setminus \mu_i)) ]\right]\notag\\
&=\EE_{\bm{r}(\mu_t)}\EE_{\bm{s}(\mu_t)}\EE_{\bm{r}(A\setminus \mu_t)}\left[r_{\mu(t)}(e)\cdot\PP\left[\calP_{\mu(t)}(\bm{r})=(S,e)\right]\prod_{i=1}^{t-1} \PP[ \calE_{\mu(i)}(\bm{r}(\mu_i\cup A\setminus \mu_t)\cup \bm{s}(\mu_t\setminus \mu_i)) ]\right],\label{sample:partitioning2}
\end{align}
where the second equality holds since the samples in $\bm{s}(A\setminus \mu_t)$ have equal distribution to those in $\bm{r}(A\setminus \mu_t)$ (and they are all independent), and by combining weight profiles (e.g., $\bm{r}=\bm{r}(\mu_t)\cup \bm{r}(A\setminus \mu_t)$), so we exchange them in the expression inside. Then, from \eqref{sample:wedge-prob} and \eqref{sample:partitioning2} we get
\begin{align}
&\EE_{\bm{s},\bm{r}}\left[\EE_{\mu}\left[r_{\mu(t)}(e)\cdot \PP\left[\calP_{\mu(t)}(\calV(\mu,t))=(S,e)\wedge \bigwedge_{i=1}^{t-1} \calP_{\mu(i)}(\calV(\mu,i)) \text{ does not block } (S,e) \right]\right]\right]\notag\\
&=\frac{1}{m!}\sum_{\mu\in \bijec([m],A)}\EE_{\bm{s},\bm{r}}\left[r_{\mu(t)}(e)\cdot\PP\left[\calP_{\mu(t)}(\bm{r})=(S,e)\right]\cdot \prod_{i=1}^{t-1}\PP[ \calE_{\mu(i)}(\bm{r}(\mu_i\cup (A\setminus \mu_t))\cup \bm{s}(\mu_t\setminus \mu_i))]\right]\notag\\
&=\EE_{\bm{s},\bm{r}}\left[ \EE_\mu \left[r_{\mu(t)}(e)\cdot\PP\left[\calP_{\mu(t)}(\bm{r})=(S,e)\right]\cdot \prod_{i=1}^{t-1} \PP[\calE_{\mu(i)}(\bm{r}(\mu_i\cup A\setminus \mu_t)\cup \bm{s}(\mu_t\setminus \mu_i)) ]\right]\right]. \label{sample:partitioning3}
\end{align}

Observe that all the numbers inside the expectation, as random variables depending on $\mu$, are not independent, so we cannot separate the expectation of the product as a product of expectations. Nevertheless, we show next how to lower-bound the expression on the inner expectation of \eqref{sample:partitioning3}. 
\begin{claim}\label{claim:key-inequality-sample}
For fixed weight profiles $\bm{r}, \bm{s}$, let $\phi_i(\mu)=\PP[\calE_{\mu(i)}(\bm{r}(\mu_i\cup A\setminus \mu_t)\cup \bm{s}(\mu_t\setminus \mu_i))]$ for each $i\in \{1,\ldots,t-1\}$. Then, for every $j\in \{1,\ldots,t-1\}$, 
\begin{align*}
&\EE_{\mu}\left[r_{\mu(t)}(e)\cdot\PP\left[\calP_{\mu(t)}(\bm{r})=(S,e)\right]\cdot \prod_{i=j}^{t-1} \phi_i(\mu) \right]\\
&\ge \left(1-\frac{k}{m-t+j}\right)_+\EE_{\mu}\left[r_{\mu(t)}(e)\cdot\PP\left[\calP_{\mu(t)}(\bm{r})=(S,e)\right]\cdot \prod_{i=j+1}^{t-1} \phi_i(\mu)\right].
\end{align*}
\end{claim}
To prove Claim \ref{claim:key-inequality-sample}, we make use of the blocking property \ref{sampler-c} of the certificate sampler $\calP$; the proof is in Appendix \ref{app:algorithms}.
By using Claim \ref{claim:key-inequality-sample} repeatedly from $j=1$ to $t-1$ we get that 
\begin{align*}
&\EE_{\mu}\left[r_{\mu(t)}(e)\cdot\PP\left[\calP_{\mu(t)}(\bm{r})=(S,e)\right]\cdot \prod_{i=1}^{t-1} \phi_i(\mu) \right]\\
&\ge \prod_{i=1}^{t-1}\left(1-\frac{k}{m-t+i}\right)_+\EE_{\mu}\left[r_{\mu(t)}(e)\cdot\PP\left[\calP_{\mu(t)}(\bm{r})=(S,e)\right]\right].
\end{align*}
To conclude the proof of the lemma, we only need to take expectation over $\bm{s}, \bm{r}$ of the previous expression and note that the right-hand side does not depend on $\bm{s}$ at all.
\end{proof}

\begin{lemma}\label{lem:sample-template-opt}
Suppose that in Algorithm \ref{alg:sample-template}, $\calP$ is a $(\gamma,k)$-certificate sampler.
Then, for $t\in \{1,\ldots,m\}$, and when $\mu$ is a random order, we have
$\EE_{\bm{r}}\left[\EE_{\mu}\left[\EE\left[r_{\mu(t)}(e_{\mu(t)}(\calP,\bm{r}))\right]\right]\right]\ge \frac{\gamma}{m}\;\EE_{\bm{r}}[\opt(\bm{r})].$
\end{lemma}

\begin{proof}
By partitioning the set $\bijec([m],A)$ according to the value $\mu(t)=a\in A$, we have that
\begin{align}
\EE_{\bm{r}}\left[\EE_{\mu}\left[\EE\left[r_{\mu(t)}(e_{\mu(t)}(\calP,\bm{r}))\right]\right]\right]&=\frac{1}{|\bijec([m],A)|}\sum_{a\in A}\sum_{\substack{\mu\in \bijec([m],A);\\ \mu(t)=a}}\EE_{\bm{r}}\left[\EE\left[r_{a}(e_{a}(\calP,\bm{r}))\right]\right]\notag\\
&=\frac{(m-1)!}{m!}\sum_{a\in A}\EE_{\bm{r}}\left[\EE\left[r_{a}(e_{a}(\calP,\bm{r}))\right]\right]\ge \frac{\gamma}{m}\;\EE_{\bm{r}}[\opt(\bm{r})],\notag
\end{align}
where the second equality holds since the summation term does not longer depend on the order $\mu$ and there are $(m-1)!$ orders such that $\mu(t)=a$, and the inequality holds by the approximation property \ref{sampler-b} of the certificate sampler $\calP$ when applied with the weight profile $\bm{r}$ for the set of agents $A$.
This finishes the proof.
\end{proof}

\begin{proof}[Proof of Theorem \ref{thm:sample-template}]
Recall that given an order $\mu$, we denote by $\mu_j$ the set $\{\mu(1),\ldots,\mu(j)\}$, for every $j\in [m]$. 
By Proposition \ref{prop:certifications}, the solution $\alg$ satisfies \ref{feasible-a}-\ref{feasible-b}; the proof is analogous to the one in Lemma \ref{lem:iid-template-correct}.
Suppose the samples defining the weight profiles $\bm{s}$ and $\bm{r}$ are fixed.
For every order $\mu$, the expected contribution in $\alg$ of the agent $\mu(t)$ is equal to
\begin{align}
& \EE\left[r_{\mu(t)}(e_{\mu(t)}(\calP,\calV(\mu,t)))\cdot \mathbb{1}\left[\bigwedge_{i=1}^{t-1} \calP_{\mu(i)}(\calV(\mu,i)) \text{ does not block } \calP_{\mu(t)}(\calV(\mu,t)) \right]\right] \notag\\
&=\hspace{-0.5cm}\sum_{(S,e)\in \mathcal{N}\cup \{(\bot,\bot)\}}\hspace{-0.7cm} r_{\mu(t)}(e)\cdot \PP\left[\calP_{\mu(t)}(\calV(\mu,t))=(S,e)\wedge \bigwedge_{i=1}^{t-1} \calP_{\mu(i)}(\calV(\mu,i)) \text{ does not block } (S,e) \right]\label{sample:transition1}
\end{align}
where the second equality holds by conditioning on the realization of $\calP(\calV(\mu,t))$.
By Lemma \ref{lem:sample-template-prod}, when $\mu$ is a random order, for every $(S,e)\in \calN\cup \{(\bot,\bot)\}$ we have
\begin{align}
&\EE_{\bm{s},\bm{r}}\left[\EE_{\mu}\left[r_{\mu(t)}(e)\cdot \PP\left[\calP_{\mu(t)}(\calV(\mu,t))=(S,e)\wedge \bigwedge_{i=1}^{t-1} \calP_{\mu(i)}(\calV(\mu,i)) \text{ does not block } (S,e) \right]\right]\right]\notag \\
&\ge \EE_{\bm{r}}\left[\EE_{\mu}\left[r_{\mu(t)}(e)\cdot\PP[\calP_{\mu(t)}(\bm{r})=(S,e)]\right]\right]\cdot \prod_{i=1}^{t-1}\left(1-\frac{k}{m-t+i}\right)_+.\label{sample:product2}
\end{align}
Then, from \eqref{sample:transition1}-\eqref{sample:product2}, the expected weight of $\alg$ can be lower bounded as
\begin{align}
& \sum_{t=1}^m\prod_{i=1}^{t-1}\left(1-\frac{k}{m-t+i}\right)_+ \sum_{(S,e)\in \mathcal{N}\cup \{(\bot,\bot)\}} \EE_{\bm{r}}\left[\EE_{\mu}\left[r_{\mu(t)}(e)\cdot\PP[\calP_{\mu(t)}(\bm{r})=(S,e)]\right]\right] \notag\\
&= \sum_{t=1}^m \prod_{i=1}^{t-1}\left(1-\frac{k}{m-t+i}\right)_+\EE_{\bm{r}}\left[\EE_{\mu}\left[\EE\left[r_{\mu(t)}(e_{\mu(t)}(\calP,\bm{r}))\right]\right]\right] \notag\\
&\ge \frac{\gamma}{m}\sum_{t=1}^m\; \prod_{i=1}^{t-1}\left(1-\frac{k}{m-t+i}\right)_+ \EE_{\bm{r}}[\opt(\bm{r})],\label{sample:transition3}
\end{align}
where the last inequality holds by Lemma \ref{lem:sample-template-opt}.
\begin{claim}\label{claim:falling}
For every positive integers $m$ and $k$, we have $\frac{1}{m}\sum_{t=1}^m\; \prod_{i=1}^{t-1}\left(1-\frac{k}{m-t+i}\right)_+\ge \frac{1}{k+1}$.
\end{claim}
The proof of Claim \ref{claim:falling} is in Appendix \ref{app:algorithms}.
By Claim \ref{claim:falling} and
using \eqref{sample:transition3}, the expected weight of $\alg$ can be lower bounded by $\gamma\;\EE_{\bm{r}}[\opt(\bm{r})]/(k+1)$, which concludes the proof of the theorem.
\end{proof}

\subsection{Secretary Model}\label{sec:sec}

In Algorithm \ref{alg:sec-template}, we provide a template that, given a certificate sampler $\calP$ with certifier $(\calS,\calN,\calB)$, constructs a solution satisfying \ref{feasible-a}-\ref{feasible-b}.

The algorithm sees the agents in $A$ according to a random order $\mu$, meaning that $\mu(1)$ is the first agent revealed, $\mu(2)$ the second agent revealed, and so on.
Recall that we denote by $\mu_t$ the set of agents $\{\mu(1),\ldots,\mu(t)\}$.
\begin{algorithm}[H]
    \begin{algorithmic}[1]
    \Require{An instance with $m$ agents $A$, revealed in random order $\mu$.}
    \Ensure{A feasible assignment of the agents.}
    \State Initialize $\alg(a)=\bot$ for every agent $a\in A$. 
    \State Let $\tau\sim \Bin(m,p)$.    
    \State {\bf Learning:} Every agent $\mu(j)$ with $j\in[\tau]$ reveals its weight function $w_{\mu(j)}$.
    \State {\bf Selection:} Initialize an empty sequence $C$ of certificates in $\calN$.
    \For{$t=\tau+1$ to $m$}
        \State The agent $\mu(t)$ reveals the weight function $w_{\mu(t)}$. 
        \State Run $\mathcal{P}$ on the weight profile $\bm{w}(\mu_t)$, where $\mu_t=\{\mu(1),\ldots,\mu(t)\}$.
        \If{$\calP_{\mu(t)}(\bm{w}(\mu_t))\neq (\bot,\bot)$ and every certificate in $C$ does not block $\calP_{\mu(t)}(\bm{w}(\mu_t))$} 
            \State Update $\alg(\mu(t))=e_{\mu(t)}(\calP,\bm{w}(\mu_t))$, and append $\calP_{\mu(t)}(\bm{w}(\mu_t))$ to the end of $C$.\label{line:sec-certification}
        \EndIf    
    \EndFor
    \State Return $\alg$.
    \end{algorithmic}
    \caption{Template for the secretary model}
    \label{alg:sec-template}
\end{algorithm}
Consider $\{(p_k,\alpha_k)\}_{k\in \NN}$ defined as follows: $(p_1,\alpha_1)=(1/e,1/e)$, and $(p_k,\alpha_k)=(k^{-\frac{1}{k-1}},k^{-\frac{k}{k-1}})$ when $k\ge 2$.
The following is the main result of this section.
\begin{theorem}\label{thm:sec-template} 
Let $(E,\calI)$ be an independence system, and let $\calP$ be a $(\gamma,k)$-certificate sampler.
Then, if we set $p=p_k$, Algorithm \ref{alg:sec-template} is $\gamma\alpha_k$-competitive in the secretary model.
\end{theorem}

In this model we have  three sources of randomness: the random order $\mu$, the random index $\tau$ and the internal randomness  of the certificate sampler $\calP$. As before, we use subindices to denote the variables over which we are taking expectations, and when we use $\EE$ or $\PP$ without subindex, it means that we are computing the expectation or the probability of an event on the probability space induced only by the internal randomness of the certificate sampler.
Before proving the theorem, in the following lemmas, we show how to use the properties of the certificate sampler $\calP$ to lower bound the weight of the assignment $\alg$ recovered by the algorithm.

\begin{lemma}\label{lem:sec-template-prod}
Consider a fixed realization of $\tau$ in Algorithm \ref{alg:sec-template}, and suppose that $\calP$ is a $(\gamma,k)$-certificate sampler.
For every $t\in \{\tau+2,\ldots,m\}$, and every $(S,e)\in \calN\cup \{(\bot,\bot)\}$ we have
\begin{align}
&\EE_{\mu}\left[w_{\mu(t)}(e)\cdot \PP\left[\calP_{\mu(t)}(\bm{w}(\mu_t))=(S,e)\wedge \bigwedge_{i=\tau+1}^{t-1} \calP_{\mu(i)}(\bm{w}(\mu_i)) \text{ does not block } (S,e) \right]\right]\notag \\
&\ge \EE_{\mu}\left[w_{\mu(t)}(e)\cdot \PP\left[\calP_{\mu(t)}(\bm{w}(\mu_t))=(S,e)\right]\right]\cdot \prod_{i=\tau+1}^{t-1}\left(1-\frac{k}{i}\right)_+.\notag
\end{align}
\end{lemma}

\begin{proof}
Given an order $\mu$, we denote by $\calC_i(\mu)$ the event in which $\calP_{\mu(i)}(\bm{w}(\mu_i))$ does not block $(S,e)$.
By expanding the expectation over the random order $\mu$, and by partitioning $\bijec([m],A)$, we get
\begin{align}
&\EE_{\mu}\left[w_{\mu(t)}(e)\cdot \PP\left[\calP_{\mu(t)}(\bm{w}(\mu_t))=(S,e)\wedge \bigwedge_{i=\tau+1}^{t-1} \calP_{\mu(i)}(\bm{w}(\mu_i)) \text{ does not block } (S,e) \right]\right]\notag\\
&=\frac{1}{m!}\sum_{a\in A}w_a(e)\sum_{\substack{T\subseteq A:\\a\in T,\\|T|=t}}\sum_{\substack{\mu \in \bijec([m],A):\\\mu(t)=a,\\\mu_{t}=T}}\PP\left[\calP_{a}(\bm{w}(T))=(S,e)\wedge \bigwedge_{i=\tau+1}^{t-1} \calC_i(\mu) \right]\notag\\
&=\frac{1}{m!}\sum_{a\in A}w_a(e)\sum_{\substack{T\subseteq A:\\a\in T,\\|T|=t}}\sum_{\substack{\mu \in \bijec([m],A):\\\mu(t)=a,\\\mu_{t}=T}}\PP\left[\calP_{a}(\bm{w}(T))=(S,e)\right]\cdot \prod_{i=\tau+1}^{t-1}\PP[ \calC_i(\mu)]\notag\\
&=\frac{1}{m!}\sum_{a\in A}w_a(e)\sum_{\substack{T\subseteq A:\\a\in T,\\|T|=t}}\PP\left[\calP_{a}(\bm{w}(T))=(S,e)\right]\sum_{\substack{\mu \in \bijec([m],A):\\\mu(t)=a,\\\mu_{t}=T}}\prod_{i=\tau+1}^{t-1}\PP[ \calC_i(\mu)],\label{sec:wedge-prob}
\end{align}
where the first equality holds by partitioning $\bijec([m],A)$ according to $\mu(t)=a$ and $\mu_t=T$; in the second equality, we use the fact that for any fixed $\mu$ all $\calP_{\mu(i)}(\bm{w}(\mu_i))$ are stochastically independent since the only randomness involved are the internal coins used in different runs of $\calP$, and the third equality holds by rearranging the summation terms.

In what follows, we fix an agent $a\in A$ and $T\subseteq A$ with $|T|=t$. 
We denote by $U(a,t,T)$ the subset of orders $\mu\in \bijec([m],A)$ such that $\mu(t)=a$ and $\mu_t=T$.
We show next how to lower-bound the right-most inner summation in \eqref{sec:wedge-prob}, which is equal to $\sum_{\mu \in U(a,t,T)}\prod_{i=\tau+1}^{t-1}\PP[\calC_i(\mu)]$.
\begin{claim}\label{claim:key-inequality}
For every $j\in \{\tau+1,\ldots,t-1\}$, we have
\begin{equation*}
\sum_{\mu \in U(a,t,T)}\prod_{i=j}^{t-1}\PP[\calC_i(\mu)]\ge \left(1-\frac{k}{j}\right)_+\sum_{\mu \in U(a,t,T)}\prod_{i=j+1}^{t-1}\PP[\calC_i(\mu)].
\end{equation*}
\end{claim}
To prove Claim \ref{claim:key-inequality}, we make use of the blocking property \ref{sampler-c} of the certificate sampler $\calP$; the proof is in Appendix \ref{app:algorithms}.
By using Claim \ref{claim:key-inequality} repeatedly from $j=\tau+1$ we get that 
$$\sum_{\mu \in U(a,t,T)}\prod_{i=\tau+1}^{t-1}\PP[\calC_i(\mu)]\ge |U(a,t,T)|\prod_{i=\tau+1}^{t-1}\left(1-\frac{k}{i}\right)_+,$$
and therefore, we can lower-bound \eqref{sec:wedge-prob} by
\begin{align*}
&\prod_{i=\tau+1}^{t-1}\left(1-\frac{k}{i}\right)_+\frac{1}{m!}\sum_{a\in A}w_a(e)\sum_{\substack{T\subseteq A:\\a\in T,|T|=t}}\sum_{\mu\in U(a,t,T)}\PP\left[\calP_{a}(\bm{w}(T))=(S,e)\right]\\
&=\prod_{i=\tau+1}^{t-1}\left(1-\frac{k}{i}\right)_+\frac{1}{m!}\sum_{a\in A}w_a(e)\sum_{\mu\in \bijec([m],A):\mu(t)=a}\PP\left[\calP_{a}(\bm{w}(\mu_t))=(S,e)\right]\\
&=\prod_{i=\tau+1}^{t-1}\left(1-\frac{k}{i}\right)_+\EE_{\mu}\left[w_{\mu(t)}(e)\cdot \PP\left[\calP_{\mu(t)}(\bm{w}(\mu_t))=(S,e)\right]\right],
\end{align*}
which finishes the proof of the lemma. 
\end{proof}

\begin{lemma}\label{lem:sec-template-opt}
Consider a fixed realization of $\tau$ in Algorithm \ref{alg:sec-template}, and suppose that $\calP$ is a $(\gamma,k)$-certificate sampler.
Then, for every $t\in \{\tau+1,\ldots,m\}$, and when $\mu$ is a random order, we have $\EE_{\mu}\left[\EE\left[w_{\mu(t)}(e_{\mu(t)}(\calP,\bm{w}(\mu_t)))\right]\right]\ge \frac{\gamma}{m}\opt(\bm{w}).$
\end{lemma}
\begin{proof}
By partitioning the set $\bijec([m],A)$, we have that
\begin{align}
\EE_{\mu}\left[\EE\left[w_{\mu(t)}(e_{\mu(t)}(\calP,\bm{w}(\mu_t)))\right]\right]&=\frac{1}{|\bijec([m],A)|}\sum_{\substack{B\subseteq A:\\|B|=t}}\sum_{a\in B}\sum_{\substack{\mu\in \bijec([m],A)
:\\\mu([t])=B,\\ \mu(t)=a}}\EE[w_a(e_a(\calP,\bm{w}(\mu_t)))]\notag\\
&=\frac{(m-t)!(t-1)!}{m!}\sum_{\substack{B\subseteq A:\\|B|=t}}\sum_{a\in B}\EE[w_a(e_a(\calP,\bm{w}(B)))],\notag\\
&\ge \frac{(m-t)!(t-1)!}{m!}\sum_{\substack{B\subseteq A:\\|B|=t}} \gamma \opt(\bm{w}(B)),\label{sec:opt-ineq1}
\end{align}
where the second equality holds since $\bm{w}(\mu_t)=\bm{w}(\mu([t]))=\bm{w}(B)$ for every $\mu\in \bijec([m],A)$ with $\mu([t])=B$ and $|\{\mu\in \bijec([m],A):\mu([t])=B,\mu(t)=a\}|=(m-t)!(t-1)!$, and the second inequality holds by the approximation property \ref{sampler-b} of the certificate sampler $\calP$ when applied with the weight profile $\bm{w}$ for the set of agents $B$.

Consider an optimal solution $M^{\star}$ in the optimization problem \eqref{eq:optimal-assignment} with weight profile $\bm{w}$. Given any $B\subseteq A$, let $M\colon B\to E\cup \{\bot\}$ such that $M(a)=M^{\star}(a)$ for every $a\in B$.
Then, $M$ satisfies \ref{feasible-a}-\ref{feasible-b} when we solve the optimization problem \eqref{eq:optimal-assignment} with weight profile $\bm{w}(B)$, and therefore, for every $B\subseteq A$, we have
\begin{equation}
\opt(\bm{w}(B))\ge \sum_{a\in B}w_{a}(M^{\star}(a)).\label{sec:opt-ineq2}
\end{equation}
Then, from \eqref{sec:opt-ineq1} and \eqref{sec:opt-ineq2}, we have
\begin{align*}
\EE_{\mu}\left[\EE\left[w_{\mu(t)}(e_{\mu(t)}(\calP,\bm{w}(\mu_t)))\right]\right]&\ge \gamma\;\frac{(m-t)!(t-1)!}{m!}\sum_{\substack{B\subseteq A:\\|B|=t}}\sum_{a\in B}w_{a}(M^{\star}(a))\\
&=\frac{\gamma/m}{\binom{m-1}{t-1}}\sum_{a\in a}w_{a}(M^{\star}(a))\cdot |\{B\subseteq A:|B|=t\text{ and }a\in B\}|\\
&=\frac{\gamma/m}{\binom{m-1}{t-1}}\sum_{a\in a}w_{a}(M^{\star}(a))\cdot |\{B\subseteq A-\{a\}:|B|=t-1\}|\\
&=\frac{\gamma/m}{\binom{m-1}{t-1}}\sum_{a\in a}w_{a}(M^{\star}(a))\cdot \binom{m-1}{t-1}=\frac{\gamma}{m}\opt(\bm{w}),
\end{align*}
where the first equality holds by rearranging the sum and the binomial coefficient, and the second holds by an equivalent counting of the number of sets while fixing the agent $a$. 
This finishes the proof.
\end{proof}

We are ready to prove Theorem \ref{thm:sec-template}.
\begin{proof}[Proof of Theorem \ref{thm:sec-template}]
Recall that given an order $\mu$, we denote by $\mu_j$ the set $\{\mu(1),\ldots,\mu(j)\}$, for every $j\in [n]$. 
By Proposition \ref{prop:certifications}, the solution $\alg$ satisfies \ref{feasible-a}-\ref{feasible-b}; the proof is analogous to the one in Lemma \ref{lem:iid-template-correct}.
Suppose that the value of $\tau$ has been realized, and let $t\geq \tau+1$.
For every order $\mu$, the expected contribution in $\alg$ of the agent $\mu(t)$ is equal to
\begin{align}
& \EE\left[w_{\mu(t)}(e_{\mu(t)}(\calP,\bm{w}(\mu_t)))\cdot \mathbb{1}\left[\bigwedge_{i=\tau+1}^{t-1} \calP_{\mu(i)}(\bm{w}(\mu_i)) \text{ does not block } \calP_{\mu(t)}(\bm{w}(\mu_t)) \right]\right] \notag\\
&=\hspace{-0.5cm}\sum_{(S,e)\in \mathcal{N}\cup \{(\bot,\bot)\}}\hspace{-0.7cm} w_{\mu(t)}(e)\cdot \PP\left[\calP_{\mu(t)}(\bm{w}(\mu_t))=(S,e)\wedge \bigwedge_{i=\tau+1}^{t-1} \calP_{\mu(i)}(\bm{w}(\mu_i)) \text{ does not block } (S,e) \right]\label{sec:transition1}
\end{align}
where the second equality holds by conditioning on the realization of $\calP(\bm{w}(\mu_t))$.
By Lemma \ref{lem:sec-template-prod}, when $\mu$ is a random order, for every $(S,e)\in \calN\cup \{(\bot,\bot)\}$ we have
\begin{align}
&\EE_{\mu}\left[w_{\mu(t)}(e)\cdot \PP\left[\calP_{\mu(t)}(\bm{w}(\mu_t))=(S,e)\wedge \bigwedge_{i=\tau+1}^{t-1} \calP_{\mu(i)}(\bm{w}(\mu_i)) \text{ does not block } (S,e) \right]\right]\notag \\
&\ge \EE_{\mu}\left[w_{\mu(t)}(e)\cdot \PP\left[\calP_{\mu(t)}(\bm{w}(\mu_t))=(S,e)\right]\right]\prod_{i=\tau+1}^{t-1}\left(1-\frac{k}{i}\right)_+.\label{sec:product2}
\end{align}
Then, the expected contribution of agent $\mu(t)$ can be lower bounded by
\begin{align}
&\sum_{(S,e)\in \mathcal{N}\cup \{(\bot,\bot)\}} \EE_{\mu}\left[w_{\mu(t)}(e)\cdot \PP\left[\calP_{\mu(t)}(\bm{w}(\mu_t))=(S,e)\right]\right]\prod_{i=\tau+1}^{t-1}\left(1-\frac{k}{i}\right)_+\notag\\
&= \EE_{\mu}\left[\EE\left[w_{\mu(t)}(e_{\mu(t)}(\calP,\bm{w}(\mu_t)))\right]\right] \prod_{i=\tau+1}^{t-1}\left(1-\frac{k}{i}\right)_+ \ge \frac{\gamma}{m}\prod_{i=\tau+1}^{t-1}\left(1-\frac{k}{i}\right)_+ \opt(\bm{w}),\label{sec:transition3}
\end{align}
where the first inequality holds from  \eqref{sec:transition1} and \eqref{sec:product2}, the equality is a consequence of the linearity of the expectation, and the second inequality holds by Lemma \ref{lem:sec-template-opt}.
Then, using \eqref{sec:transition3}, and when $\tau\sim \Bin(m,p_k)$, the expected weight of the solution $\alg$ computed by the algorithm from $\tau+1$ up to $m$ can be lower bounded as
\begin{align*}
& \gamma \opt(\bm{w})\cdot \EE_{\tau}\left[\frac{1}{m}\sum_{t=\tau+1}^m\prod_{i=\tau+1}^{t-1}\left(1-\frac{k}{i}\right)_+\right]\ge \gamma\alpha_k \opt(\bm{w}),
\end{align*}
where the inequality holds by \cite[Lemma 1, p. 8]{SotoTV2021}. This finishes the proof.
\end{proof}

\section{Guarantees via Polynomial-time Certificate Samplers}\label{sec:polytime-certifiers}
In this section, we show how to design polynomial-time certificate samplers for several independence systems, including hypergraph matchings, matroids, and matroid intersections. 
Together with our algorithmic templates from Section \ref{sec:algorithms}, we can provide new and improved guarantees for the online combinatorial assignment problem in independence systems.
Recall that the $k$-bounded combinatorial auction problem is captured by problem \eqref{eq:optimal-assignment} in $k$-hypergraph matchings.
In Section \ref{sec:polytime-matching}, we show a polynomial-time certificate sampler, based on linear programming, for this independence system.
Using this, we get the following guarantees for $k$-bounded combinatorial auctions.

\begin{theorem}\label{thm:dbca}
For the online $k$-bounded combinatorial auction problem, the following holds:
\begin{enumerate}[itemsep=0pt,label=(\alph*)]
    \item There exists a $(1-e^{-k})/k$-competitive algorithm in the prophet IID model.\label{dbca-iid}
    \item There exists a $1/(k+1)$-competitive algorithm in the prophet-secretary model, using a single sample per agent.\label{dbca-ps}
    \item There exists a $k^{-k/(k-1)}$-competitive algorithm in the secretary model for every $k\ge 2$ and a 1/e-competitive algorithm for $k=1$.\label{dbca-sec}
\end{enumerate}
Furthermore, all these algorithms run in polynomial time.
\end{theorem}

We prove Theorem \ref{thm:dbca} in Section \ref{sec:polytime-matching}.
Then, in Section \ref{sec:matroids}, we construct certificate samplers for matroids, defining the new concept of {\it $k$-directed certifier}. 

\begin{definition}
A tuple $(\mathcal{S},\mathcal{N},\mathcal{B})$ for an independence system $(E,\calI)$ is a $k$-directed certifier if it satisfies conditions \ref{cert-a}, \ref{cert-b}, and \ref{cert-c} in Definition \ref{def:certifier}, and the following extra conditions also hold:
%(we add them to conditions \ref{cert-a}, \ref{cert-b}, \ref{cert-c})
\begin{enumerate}[itemsep=0pt,label=(\alph*),start=4] 
    \item $\mathcal{S}=\mathcal{I}$ and $\mathcal{N}=\{(I,e)\colon I\in \mathcal{I}, e\in I\}$.\label{cert-d} 
    \item For every $I, J \in \mathcal{I}$ and every $f\in J$, we have $|\{e\in I\colon ((I,e),(J,f))\in \mathcal{B}|\leq k$. \label{cert-e} 
\end{enumerate}
\end{definition}

The concept of $k$-directed certifier is very similar to the $k$-forbidden technique developed by Soto et al. \cite{SotoTV2021} for the matroid secretary problem. Most of the $k$-forbidden algorithms in \cite{SotoTV2021} can be implemented by finding an appropriate directed certifier for the matroid and then running our algorithm. 
In Section \ref{sec:matroids}, we construct certificate samplers for matroids admitting a $k$-directed certifier.
Combined with our algorithmic framework from Section \ref{sec:algorithms}, we obtain the following result.
\begin{theorem}\label{thm:matroids}
Let $\calM$ be a matroid admitting a $k$-directed certifier.  
For the online combinatorial assignment problem in $\calM$  the following holds:
\begin{enumerate}[itemsep=0pt,label=(\alph*)]
    \item There exists a $(1-e^{-k})/k$-competitive algorithm in the prophet IID model.\label{matroid-iid}
    \item There exists a $1/(k+1)$-competitive algorithm in the prophet-secretary model, using a single sample per agent.\label{matroid-ps}
    \item There exists a $k^{-k/(k-1)}$-competitive algorithm in the secretary model for every $k\ge 2$ and a $1/e$-competitive algorithm for $k=1$.\label{matroid-sec}
\end{enumerate}
Furthermore, all these algorithms run in polynomial time.
\end{theorem}
The proof of Theorem \ref{thm:matroids} is in Section \ref{sec:matroids}.
In particular, Theorem \ref{thm:matroids} holds for unitary matroids $(k=1)$, graphic, transversal and matching matroids $(k=2)$, and they hold for general $k$ in $k$-sparse matroids, $k$-framed matroids, $k$-exchangeable gammoids, and $k$-exchangeable matroidal packings as defined in \cite{SotoTV2021}.

In Section \ref{sec:combinations} we discuss certificate samplers for {\it matchoid} independence systems \cite{lovasz1980matroid,jenkyns1975matchoids}, which generalize matchings and matroids.
Given $(E_i,\calI_i)_{i\in [t]}$ an arbitrary collection of matroids whose ground sets do not necessarily coincide, the \emph{matchoid} associated to the collection is the independence system $(E,\calI)$ with $E=\bigcup_{i\in [t]}E_i$ and such that $X\subseteq E$ is independent if and only if for all $i\in [t]$, $X\cap E_i\in \calI_i$.
Note that if all $E_i$ are the same, this is simply the intersection of the $t$ matroids. This construction is usually called an $\ell$-matchoid if every element $e\in E$, is active in at  most $\ell$ of the matroids (i.e., $e$ belongs to at most $\ell$ sets $E_i$). Just like matching on bipartite graphs can be modeled as the intersection of 2 matroids, matchings on a general $k$-hypergraph $G=(V,E)$ can be modeled as a $k$-matchoid. Indeed, we can define a unitary matroid $(E_v,\calI_v)$ for every vertex $v\in V$ where $E_v=\{e\in E\colon e \text{ is incident to } v\}$ and $X\subseteq E_v$ is a member of $\calI_v$ if $|X|\leq 1$. Since every edge of the hypergraph has at most $k$ endpoints, this construction is indeed a $k$-matchoid.
We show the following guarantee.

\begin{theorem}\label{thm:matroid-intersection}
Let $(E,\calI)$ be an independence system obtained as either the intersection of $\ell$ matroids $(E,\calI_1),\ldots,  (E,\calI_\ell)$ where matroid $(E,\calI_i)$ admits a $k_i$-directed certifier, or as a matchoid such that for each $e\in E$, the $t(e)$ matroids involving element $e$ in the matchoid admit a directed certifier with parameters $k_1(e),\dots, k_{t(e)}(e)$.
Let $k=\sum_{i=1}^{\ell}k_i$ in the first case, and $k=\textstyle\max_{e\in E}\sum_{i=1}^{t(e)}k_i(e)$ in the second case. Then, for the online combinatorial assignment problem in $(E,\calI)$  the following holds:
\begin{enumerate}[itemsep=0pt,label=(\alph*)]
    \item There exists a $(1-e^{-k})/k$-competitive algorithm in the prophet IID model.\label{mi-iid}
    \item There exists a $1/(k+1)$-competitive algorithm in the prophet-secretary model, using a single sample per agent.\label{mi-ps}
    \item There exists a $k^{-k/(k-1)}$-competitive algorithm in the secretary model for every $k\ge 2$ and a $1/e$-competitive algorithm for $k=1$.\label{mi-sec}
\end{enumerate}
Furthermore, all these algorithms run in polynomial time.
\end{theorem}

The proof of Theorem \ref{thm:matroid-intersection} is in Section \ref{sec:combinations}.

\subsection{Certificate Samplers for Hypergraph Matchings}\label{sec:polytime-matching}

Let $(E,\calI)$ be the system of $k$-hypergraph matchings in $G=(V,E)$.
We consider the certifier constructed in Example \ref{example:hypergraph}, i.e., $\calS=E$, $\mathcal{N}=\{(e,e)\colon e\in E\}$, and $\calB=\{((e,e),(f,f))\colon e\cap f\ne \emptyset\}$.
Namely, every edge defines its own certificate, and we have that $(e,e)$ blocks $(f,f)$ if $e$ and $f$ have at least one node in common.
Given a weight profile $\bm{w}$ for $A$, consider the following linear program: 
\begin{center}
\begin{minipage}{0.7\textwidth}
\begin{align}
\max \quad\;  \sum_{a\in A}\sum_{e \in E} w_a(e) &x_{a}(e) \tag*{\mbox{HM$(\bm{w})$}} \label{lp-hm} \\
\text{s.t.} \quad  \sum_{a\in A}\sum_{e \in \delta(v)} x_{a}(e) &\leq 1 \quad \text{ for every } v \in V, \notag\\
 x_{a}(\bot)+\sum_{e \in E} x_{a}(e) &= 1 \quad \text{ for every } a\in A,\notag\\
 x_{a}(e) &\geq 0 \quad \text{ for every } a\in A,\text{ and every } e \in E\cup \{\bot\}.\notag
\end{align}
\end{minipage}
\end{center}
\vspace{.3cm}

In the linear program \ref{lp-hm}, we have one variable $x_{a}(e)$ for each $a\in A$ and each $e\in E\cup \{\bot\}$, and it models whether $e\in E\cup \{\bot\}$ is assigned to $a\in A$.
This linear program corresponds to a linear relaxation of the problem \eqref{eq:optimal-assignment} over the $k$-hypergraph matchings of $G$.\footnote{We assume that the linear program has a unique optimal solution, e.g., solved by perturbing the weights lexicographically using names of edges and weight functions.}
The first set of constraints ensures that in every node we have a total fractional allocation of at most one, and the second set of constraints encodes that every $a\in A$ has a probability distribution over $E\cup \{\bot\}$.
Consider the following randomized algorithm.
\begin{algorithm}[H]
    \begin{algorithmic}[1]
    \Require{A weight profile for $A$}
    \Ensure{A collection of certificates in $\calN$}
    \State Find the unique optimal solution $x^{\star}$ of the linear program \ref{lp-hm}.
    \State For every $a\in A$, sample $e_a \in E\cup \{\bot\}$ according to the probability distribution $(x^{\star}_a(e))_{e\in E\cup \{\bot\}}$.
    \State Return $(e_a,e_a)$ for every $a\in A$. 
    \end{algorithmic}
    \caption{Certificate sampler for hypergraph matchings}
    \label{alg:sampler-matchings}
\end{algorithm}

\begin{lemma}\label{lem:sampler-matching}
Algorithm \ref{alg:sampler-matchings} is a $(1,k)$-certificate sampler for $k$-hypergraphs matchings, and it runs in polynomial time.
\end{lemma}
\begin{proof}
It is clear that the algorithm runs in polynomial time on a weight profile $\bm{w}$.
Note that 
   \begin{align*}
       \sum_{a\in A} \EE[w_a(e_a)]&=\sum_{a\in A}\sum_{e\in E\cup \{\bot\}} w_a(e)x^{\star}_a(e)\ge \sum_{a\in A}\sum_{e\in E} w_a(e)x^{\star}_a(e)\geq \opt(\bm{w}),
   \end{align*}
   where we have used that $w_a(\bot)=0$ for all $a$ and that the \ref{lp-hm} is a relaxation of the maximum weight feasible assignment problem \eqref{eq:optimal-assignment}, so its value is at least $\opt(\bm{w})$.  
   Thus, condition \ref{sampler-b} in Definition \ref{def:good-sampler} holds with $\gamma=1$. 
   Now let $(f,f)$ be any certificate. We have that
   \begin{align*}
    \sum_{a\in A}\PP[((e_a,e_a),(f,f))\in \calB]&= \sum_{a\in A} \sum_{e\in E: e\cap f\ne \emptyset } x_{a}(e)\\
    &\leq \sum_{a\in A}\sum_{v\in V: v\in f} \sum_{e\in \delta(v)} x_{a}(e)\leq \sum_{v\in V: v\in f} \sum_{a\in A}\sum_{e\in \delta(v)} x_{a}(e)\leq |f|\le k,
    \end{align*}
    where the first inequality holds by decomposing the inner summation over each node incident to the edge $f$, the second holds by exchanging the summation order,
    the third inequality holds by the first set of constraints in the linear program \ref{lp-hm}, and the last inequality follows since every edge in the hypergraph has size at most $k$.
    Therefore, condition \ref{sampler-c} in Definition \ref{def:good-sampler} holds.
\end{proof}

\begin{proof}[Proof of Theorem \ref{thm:dbca}]
The proof holds directly by calling the $(1,k)$-certificate sampler $\calP$ of Lemma \ref{lem:sampler-matching}: The guarantee \ref{dbca-iid} in the IID model holds by using Theorem \ref{thm:iid-template}, the guarantee \ref{dbca-ps} in the prophet-secretary model holds by using Theorem \ref{thm:sample-template}, and the guarantee \ref{dbca-sec} in the secretary model holds by using Theorem \ref{thm:sec-template}.
\end{proof}
\subsection{\texorpdfstring{$k$}{k}-directed Certifiers for Matroids }\label{sec:matroids}

In this section we consider the case of independence systems admitting $k$-directed certifiers. The main result of this section is as follows. Let $(\mathcal{S},\mathcal{N},\mathcal{B})$ be a $k$-directed certifier for an independence system $(E,\calI)$. 
Consider the following deterministic algorithm:

\begin{algorithm}[H]
    \begin{algorithmic}[1]
    \Require{A weight profile $\bm{w}$ for $A$}
    \Ensure{A collection of certificates in $\calN$}
    \State Find an optimal solution $M^{\star}\colon A \to E\cup \{\bot\}$ of the combinatorial assignment problem \eqref{eq:optimal-assignment}. Let $I(\bm{w})=\{M^{\star}(a)\colon a\in A\}$ be the associated independent set.
    \State For every $a\in A$ with $M^{\star}(a)=\{\bot\}$ return $\mathcal{P}_a=(\bot,\bot)$, and for every $a\in A$ with $M^{\star}(a)\in I(\bm{w})$, return $\mathcal{P}_a=(I(\bm{w}),M^{\star}(a))$. 
    \end{algorithmic}
    \caption{Certificate sampler for $k$-directed certifiers}
    \label{alg:sampler-kdirected}
\end{algorithm}

\begin{lemma}\label{lem:sampler-k-directed}
Algorithm \ref{alg:sampler-kdirected} is a $(1,k)$-certificate sampler for $(E,\calI)$ whenever $(\mathcal{S},\mathcal{N},\mathcal{B})$ is a $k$-directed certifier. Furthermore, if the combinatorial assignment problem \eqref{eq:optimal-assignment} in $(E,\calI)$ can be solved in polynomial time, then the certificate sampler also runs in polynomial time.
\end{lemma}
\begin{proof}
Recall that from our notation in Definition \ref{def:certificate-sampler}, for every $a$ with $M^{\star}(a)\in I(\bm{w})$ we have $e_a(\calP,\bm{w})=M^{\star}(\bm{w})$.
Then, we have
$\sum_{a\in A}w_a(e_a(\mathcal{P},\bm{w}))=\sum_{a\in A} w_a(M^*(a))=\OPT(\bm{w}),$
and so condition \ref{sampler-b} in Definition \ref{def:good-sampler} holds with $\gamma=1$.
Now, let $(J,f)\in \mathcal{N}$ be any certificate. Since the algorithm is deterministic, and $(\mathcal{S},\mathcal{N},\mathcal{B})$ is a $k$-directed certifier, we have 
\begin{align*}
\sum_{a\in A}\PP[(\mathcal{P}_a(\bm{w}),(J,f))\in \mathcal{B}] &= \sum_{e\in I(\bm{w})} \mathbb{1}[((I(\bm{w}),e),(J,f))\in \mathcal{B}]\\
&= |\{e\in I(\bm{w})\colon ((I(\bm{w}),e),(J,f))\in \mathcal{B}\}| \leq k,
\end{align*}
and therefore condition \ref{sampler-c} in Definition \ref{def:good-sampler} holds.
\end{proof}

Given a matroid $\mathcal{M}=(E,\calI)$ and a weight profile $\bm{w}$ for a set $A$ of agents, we can solve in polynomial time the combinatorial assignment problem \eqref{eq:optimal-assignment} for $\bm{w}$ by modeling it as a matroid intersection. 
Indeed, the independence system $\mathcal{M}_A$ with ground set $E\times A$, where sets $F=\{(f_1,a_1), \dots (f_s,a_s)\}$ are independent if and only if $(a_1,\dots, a_s)$ are all different, and $(f_1,\dots, f_s)$ are all different with $\{f_1,\dots, f_s\}\in \calI$ is the intersection of a partition matroid associated to the agents, and a parallel extension of the matroid $\mathcal{M}$, which is also a matroid. For any maximum weight independent set $F^*=\{(f_1,a_1), \dots (f_s,a_s)\}$ of $\mathcal{M}_A$, the assignment $M\colon A\to E\cup \{\bot\}$ given by $M(a)=f_i$ if $a=a_i$ for $i\in [s]$ and $M(a)=\bot$ if $a\in A\setminus \{a_1,\dots, a_s\}$ is optimal for the combinatorial assignment problem \eqref{eq:optimal-assignment}. 
Below we describe $k$-directed certifiers for some basic matroids. 

\noindent{\bf Unitary partition matroids.} Consider the partition matroid $(V,\calI)$ defined by a collection of disjoint nonempty parts $V_1,\dots, V_\ell$. The ground set is $V=\bigcup_{i=1}^\ell V_i$. A set $I\subseteq A$ is independent if and only if $|I\cap A_i|\leq 1$ for all $i\in [\ell]$. 
We claim that $(\mathcal{S},\mathcal{N},\mathcal{B})$ given by $\calS=\calI$, $\calN=\{(I,v)\colon I\in \calI, v\in I)$ and $\calB=\{((I,v),(I',v'))\colon$
$\text{$v$ and $v'$ are in the same $V_i$}\}$ is a $1$-directed certifier. Indeed, conditions \ref{cert-a}, \ref{cert-b}, \ref{cert-d} are satisfied by construction. Each certification $(I_1,v_1),\ldots,(I_t,v_t)$ is such that for every $1\le i<j\le t$ the elements $v_i$ and $v_j$ are in different sets of $\{V_q\}_{q\in \ell}$. Therefore, $\{v_1,\dots, v_t\}$ contains elements from different parts, and so, it is independent. Thus, property \ref{cert-c} holds.

Finally, for every $I, I'\in \mathcal{I}$ and every element $v'\in I'$, we have 
$|\{v\in I\colon ((I,v),(I',v'))\in \mathcal{B}\}|=
|\{v\in I\colon v \text{ is in the same part as $v'$}\}| \leq 1,$
since every $I$ contains at most one element from each part. So part \ref{cert-e} holds and $(\mathcal{S},\mathcal{N},\mathcal{B})$ is a 1-directed certifier.

\noindent{\bf Graphic matroids.}  The graphic matroid associated with a loopless graph $G=(V,E)$ is the independent system $(E,\mathcal{I})$ where $F\subseteq E$ is an independent in $\calI$ if $F$ is acyclic. Let  $\text{or}(F)\subseteq V\times V$ be the orientation of $F$ obtained by labeling the nodes in $[n]$ as $v_1,v_2,\ldots,v_n$, rooting each connected component $C$ of $(V,F)$ (which are trees) on the node $r(C)$ with smallest label, and orient all edges of $C$ away from $r(C)$. For each $e\in F$, let $(e^-_F,e^+_F)\in \text{or}(F)$ its orientation. Consider $\calS=\calI$ and $(\mathcal{N},\calB)$ be the digraph defined as follows: $\mathcal{N}=\{(F,e)\colon F\in \calI, e\in F\},$ and $\calB=\{(F,e),(F',e'))\colon e^+_F \text{ is a node of } e'\}\}$.
Properties \ref{cert-a}, \ref{cert-b} and \ref{cert-d} are satisfied by construction. For every certification $(F_1,e_1),\ldots, (F_t,e_t)$, and every $i$, the \emph{head} of edge $e_i$, i.e. the node $(e_i)^+_{F_i}$, is not covered by any of the following edges of the sequence, and then $e_i$ cannot close a cycle with $\{e_{i+1},\dots, e_{t}\}$. By induction, this implies that $\{e_1,\dots, e_t\}$ is acyclic and, therefore, $\{e_1,\ldots,e_t\}\in \calI$. This shows property \ref{cert-c}. 

Finally, for every $F, F'\in \calI$ and every edge $e'\in F'$, we have
$|\{e\in F\colon ((F,e),(F',e'))\in \mathcal{B}\}|=|\{e\in F\colon e^+_F \text{ is a node of } e'\}|\leq 2,$
since in the orientation $\text{or}(F)$, every node $v$ has indegree at most 1. Hence, for every node $v$ of $e'$ there is at most one arc oriented toward $v$ in $\text{or}(F)$. The value 2 appears since $e'$ has two endpoints. This shows that $(\calS,\calN,\calB)$ is a 2-directed certifier.

\noindent{\bf Transversal matroids and matching matroids.} For transversal matroids,  \cite{SotoTV2021}, gave a 1-forbidden algorithm based on the bipartite graph representation of the matroid: the transversal matroid associated to the bipartite graph $(L\cup R, E)$ is the system $(L,\mathcal{I})$ where $X\subseteq L$ is independent if and only if there exist a matching $M_X$ covering all vertices in $X$. But the natural certifier $(\calS,\calN,\calB)$ given by $\calS=\calI$, $\calN=\{(X,v)\colon X\in \calI, v\in X\}$ and $\calB=\{((X,v),(X',v'))\colon M_X(v)\cap M_{X'}(v')\neq \emptyset\}$, where $M_X(v)$ (resp. $M_{X'}(v')$) is the unique edge in $M_X$ covering $v$ (resp. in $M_{X'}$ covering $v'$) is only a 2-directed system. This difference is explained because in the single-minded secretary problem over transversal matroids (which is the problem studied in \cite{SotoTV2021}) the vertices on the right hand side of the graph may appear multiple times as part of a potential edge to be added, but the vertices of the left hand side ($v\in L$) can only appear once (they only appear at the moment the agent associated to $v$ appears). However, in the combinatorial assignment version, the node $v\in L$ may also appear multiple times, as different agents may assign high weight on them.
We note that a similar construction as above works for matching matroids in which the ground set is the vertex set of a graph and the independent sets are those sets of vertices that can be covered by a matching of the graph. So, both matching and transversal matroids, admit 2-directed certifiers.

\noindent{\bf Other matroids.} By adapting the constructions of \cite{SotoTV2021}, we can also get $k$-directed certifiers for $k$-sparse matroids, $k$-framed matroids, $k$ exchangeable gammoids and $k$ exchangeable matroidal packings. However the constructions used for laminar and semiplanar matroids cannot be used directly as they require to construct a $k$-directed certifier on the fly using information revealed during the agents' arrival process. Even though our framework allows for this type of constructions we won't seek them in this article.

\begin{proof}[Proof of Theorem \ref{thm:matroids}]
The proof is exactly the same as that of Theorem \ref{thm:dbca} but using the $(1,k)$-certificate sampler guaranteed by Lemma \ref{lem:sampler-k-directed}  instead.
\end{proof}

\subsection{Combining Certifiers: Matroid Intersection and Matchoids}\label{sec:combinations}

In this section, we consider the online combinatorial assignment problem in the matchoid $(E,\calI)$ obtained from the matroids $(E_i,\calI_i)_{i\in [\ell]}$, when the $i$-th matroid admits a $k_i$-directed certifier $(\calI_i,\calN_i,\calB_i)$. 
While it is possible to construct a certifier and a certificate sampler for $(E,\calI)$ achieving the sought guaranteed by intersecting (in an appropriate way) the certificates on each of the matroids, in general, this construction does not yield polynomial-time algorithms, since the combinatorial assignment problem for the matchoid is in general $\NP$-hard. 
Instead of intersecting the certificates, we will put them together in a matchoid certifier $(\calS,\calN,\calB)$ defined as follows: $\calS=\prod_{i=1}^{\ell} \calI_i,$
\begin{align*}
\calN&=\{(I_1,\dots, I_{\ell}; e)\in \calS\times E \colon \forall i\in [\ell]\, (e\in E_i \Rightarrow e\in I_i\in \calI_i) \text{ and } (e\not\in E_i \Rightarrow I_i=\emptyset)\}, \\
\calB&=\{((I_1,\dots, I_{\ell}; e),(I'_1,\dots, I'_{\ell}; e'))\in \calN\times \calN \colon \exists i\in [\ell]\,  e\in E_i, e'\in E'_i, ((I_i,e),(I'_i,e'))\in \calB_i\}.
\end{align*}

In other words, the certificates for any $e\in E$ on the matchoids consists of putting together the certificates $(I_i,e_i)$ for each of the matroids in a certificate bundle $(I_1,\dots, I_\ell;e)$, except that whenever $e$ is not active in matroid $i$, we replace $I_i$ by $\emptyset$. Furthermore a certificate bundle $(I_1,\dots, I_\ell;e)$ blocks another certificate bundle $(I'_1,\dots, I'_\ell;e')$ if there is a matroid $(E_i,\calI)$ where both $e$ and $e'$ are active and the certificate $(I_i,e)$ for $e$ blocks the certificate $(I'_i,e')$ for $e'$.

\begin{lemma} The triple $(\calS,\calN,\calB)$ defined above is indeed a certifier for the matchoid $(E,\calI)$    
\end{lemma}
\begin{proof}
Property \ref{cert-a} holds by construction. Consider certificate bundles $(I_1,\dots, I_\ell;e)$, $(J_1,\dots, J_\ell; e) \in \calN$ for the same element $e\in E$. Let $i$ be any index such that $e\in E_i$. Then, by definition of $\calN$, $e\in I_i \in \calI_i$ and $e\in J_i \in \calI_i$. But since $(\calI_i,\calN_i,\calB_i)$ is a $k_i$-directed certifier,  properties \ref{cert-d} and \ref{cert-b} for those certifiers implies that $(I_i,e)\in \calN_i$ and $(J_i,e)\in \calN_i$, and then $((I_i,e),(J_i,e))\in \calB_i$. This shows that $(E,\calI)$ satisfies property \ref{cert-b}. 
Let now $(I^{1}_1,\dots, I^{1}_\ell;e_1), (I^{2}_1,\dots, I^{2}_\ell; e_2), \dots , (I^{t}_1,\dots, I^{t}_\ell; e_t)$ be a certification. Let $i\in [\ell]$ be any index in which, and let $I^{s_1}_i, \dots, I^{s_q}_i$ be the subsequence of $I^1_i,\dots, I^{t}_i$ obtained by keeping only the nonempty sets. Since  the original sequence is a certification in $(\calS,\calN,\calB)$, then $(I^{s_1}_i,e_{s_1}),(I^{s_2}_i,e_{s_2}),\dots (I^{s_q}_i,e_{s_q})$ is also a certification in $(\calS_i,\calN_i,\calB_i)$. Then $\{e_1,\dots, e_t\}\cap E_i= \{e_{s_1},\dots, e_{s_q}\}\in \calI_i$. Since this is true for all $i$, we get that $\{e_1,\dots, e_t\}\in \calI$ and thus, property \ref{cert-c} holds.
\end{proof}

To construct the certificate sampler for $(E,\calI)$, we use a linear program. For each $i\in [\ell]$, let $P(\calI_i)=\{x\in \RR^{E_i}_+ \colon x(S)\leq r_i(S)\; \text{for all } S\subseteq E\}=\text{conv}(\chi^{I}\colon I\in \calI_i)$ be the independence polytope associated to matroid $(E_i,\calI_i)$, where $r_i$ is the rank function of that matroid. Given a weight profile $\bm{w}$ for $A$, consider the following linear program:

\begin{center}
\begin{minipage}{0.7\textwidth}
\begin{align}
\max \quad \sum_{a\in A}\sum_{e \in E} w_a(e) &x_{a}(e) \tag*{\mbox{Matchoid$(\bm{w})$}} \label{lp-matchoid} \\
\text{s.t.} \quad\quad  \sum_{a\in A}x_{a}(e)&= y(e) \quad\;\; \text{ for every } e\in E, \notag\\
y|_{E_i}&\in P(\calI_i)  \quad \text{ for every } i \in [\ell], \notag\\
 x_{a}(\bot)+\sum_{e \in E} x_{a}(e) &= 1 \quad\quad\;\;\; \text{ for every } a\in A,\notag\\
 x_{a}(e) &\geq 0 \quad\quad\;\;\; \text{ for every } a\in A,\text{ and every } e \in E\cup \{\bot\}.\notag
\end{align}
\end{minipage}
\end{center}

It is well known that we can solve the separation problem on each $P(\calI_i)$ in polynomial time having only oracle access to $\calI_i$ \cite{cunningham1984testing}. Then, we can also solve \ref{lp-matchoid} using the ellipsoid method \cite{grotschel2012geometric}. In fact, we can assume (by perturbing infinitesimally the weight functions using names of elements and/or weight functions) that this linear program has a unique solution.
Furthermore, for every $i\in [\ell]$ and every vector $z\in P(\calI_i)$ we can, in polynomial time, decompose $z$ as a convex combination of indicator vectors of independent sets in $\calI_i$, that is $z=\sum_{I\in \calI}\lambda_I \chi^{I}$, where $\lambda_I\geq 0$ for all $I$,  $\sum_{I\in \calI}\lambda_I=1$ and $\lambda_I$ is not zero for at most a polynomial (in $|E_i|$) number of sets. 
Consider now the following randomized algorithm.
\begin{algorithm}[H]
    \begin{algorithmic}[1]
    \Require{A weight profile $\bm{w}$ for $A$}
    \Ensure{A collection of certificates in $\calN$}
    \State Find the unique optimal solution $(x^*,y^*)$ of the linear program \ref{lp-matchoid}.
    \State For each $i\in [\ell]$, decompose $y^*|_{E_i}=\sum_{I\in \calI_i}\lambda^{(i)}_I \chi^{I}$ as a convex combination of indicator vectors of sets in $\calI_i$.
    \State For every $a\in A$, sample $e_a\in E\cup \{\bot\}$ according to the probability distribution $(x_a^*(e))_{e\in E\cup \{\bot\}}$. Then, for every $i\in [\ell]$ with $e_a\not\in E_i$, set $I^{a}_i=\emptyset$ and for every $i\in [\ell]$ with $e_a\in E_i$, sample $I^{a}_i$ from $\calI_i$ according to the probability distribution $\mu^{(i,e_a)}$ given by $$\mu^{(i,e_a)}_I=\frac{\lambda^{(i)}_I}{ \sum_{J\in \calI_i\colon e_a\in J} \lambda^{(i)}_J}\text{ for }I\in \calI_i\text{ with } e_a\in I.$$
    \State Return $(I_1^a,I_2^a,\dots, I_\ell^a;e_a)$ for each $a\in A$.
    \end{algorithmic}
    \caption{Certificate sampler for matchoids}
    \label{alg:sampler-matchoid}
\end{algorithm}

\begin{lemma}\label{lem:sampler-matchoid}
For each $e\in E$, let $k(e)=\sum_{i\in [\ell]\colon e\in E_i} k_i$, where $k_i$ is such that the matroid $(E_i,\calI_i)$ admits a $k_i$-directed certifier. And let $k=\max_{e\in E} k(e)$.
Algorithm \ref{alg:sampler-matchoid} is a $(1,k)$-certificate sampler for the matchoid $(E,\calI)$ formed from matroids $(E_i,\calI_i)_{i\in [\ell]}$, and it runs in polynomial time.
\end{lemma}
\begin{proof}
By the discussion before this lemma, it is clear that the algorithm runs in polynomial time on a weight profile $\bm{w}$.
Observe that
   \begin{align*}
       \sum_{a\in A} \EE[w_a(e_a)]&=\sum_{a\in A}\sum_{e\in E\cup \{\bot\}} w_a(e)x^{\star}_a(e)\ge \sum_{a\in A}\sum_{e\in E} w_a(e)x^{\star}_a(e)\geq \opt(\bm{w}),
   \end{align*}
   where we have used that $w_a(\bot)=0$ for all $a$ and that the \ref{lp-matchoid} is a relaxation of the problem \eqref{eq:optimal-assignment}, so its value is at least $\opt(\bm{w})$.  
   Thus, condition \ref{sampler-b} holds with $\gamma=1$. 
   Now let $(J_1,J_2,\dots, J_\ell;f)\in \calN$ be any certificate. By a union bound,
   \begin{align}
    &\sum_{a\in A}\PP[((I_1^a,I_2^a,\dots,I_\ell^a;e_a),(J_1,J_2,\dots,J_\ell;f))\in \calB]\notag \\
    &= \sum_{a\in A}\PP[\text{Exists } i\in [\ell], e_a\in E_i, f\in E_i, ((I_i^a,e_a),(J_i,f))\in \calB_i]\notag  \\
    &\leq \sum_{a\in A}\; \sum_{i\in [\ell]\colon f\in E_i}\PP[e_a\in E_i, ((I_i^a,e_a),(J_i,f))\in \calB_i]\notag 
 \\
    &=\sum_{a\in A}\; \sum_{i\in [\ell]\colon f\in E_i} \; \sum_{e\in E_i} \sum_{I\in \calI_i} \PP[e_a=e \wedge I_i^a=I]\cdot \mathbb{1}[((I,e),(J_i,f))\in \mathcal{B}_i]\cdot \mathbb{1}[e\in I]. \label{eqn:largamatchoid}
    \end{align}
Recall that $y^*|_{E_i}=\sum_{I\in \calI_i}\lambda^{(i)}_I\chi^I$ and thus, for every $e\in E_i$, $y^*(e)=\sum_{I\in \calI_i\colon I\ni e}\lambda^{(i)}_I$. Furthermore, and every $I\in \calI_i$ such that $e\in I$, 
$\PP[e_a=e \wedge I_i^a=1]=x^*_a(e)\mu_I^{(i,e)}=x^*_a(e)\lambda^{(i)}_I/y^*(e).$ Using this and changing the order of summations, we have that  \eqref{eqn:largamatchoid} equals
\begin{align*}
    &\sum_{i\in [\ell]\colon f\in E_i} \sum_{I\in \calI_i}\sum_{e\in E}\frac{\lambda^{(i)}}{y^*(e)}\cdot \mathbb{1}[((I,e),(J_i,f))\in \mathcal{B}_i]\cdot \mathbb{1}[e\in I]\cdot \sum_{a\in A} x_a^*(e)\\
    &=\sum_{i\in [\ell]\colon f\in E_i} \sum_{I\in \calI_i}\lambda^{(i)} \sum_{e\in E}  \mathbb{1}[((I,e),(J_i,f))\in \mathcal{B}_i]\cdot\mathbb{1}[e\in I]\\
    &=\sum_{i\in [\ell]\colon f\in E_i} \sum_{I\in \calI_i}\lambda^{(i)}_I |\{e\in I\colon ((I,e),(J_i,f))\in \calB_i\}|\leq \sum_{i\in [\ell]\colon f\in E_i} \hspace{-.3cm}k_i = k(f) \leq k,
\end{align*}
where the first equality holds since $y^*(e)=\sum_{a\in A} x^*_a(e)$, and the last inequality holds since $\sum_{I\in \calI_i}\lambda_I=1$ and $(\calI_i,\calN_i,\calB_i)$ is a $k_i$ direct certifier. This finishes the proof. 
\end{proof}
\begin{proof}[Proof of Theorem \ref{thm:matroid-intersection}]
The proof is exactly the same as that of Theorem \ref{thm:dbca} but using the $(1,k)$-certificate sampler guaranteed by Lemma \ref{lem:sampler-matchoid}  instead.
\end{proof}

\section{Table of Known and New Results}\label{sec:table}

In Table \ref{big-table}, we provide a panorama of all the lower and upper bounds known for each online model, both for the single-minded version and the general version.
Our new results are in \textbf{bold} font. They appear in the format ($\textbf{new} \gets \text{old}$) if an upper bound is decreased of ($\text{old} \to \textbf{new}$) if a lower bound is increased. Going up or left in the table makes the problem \emph{easier} for the algorithm, so the quantities shown decrease. So any lower bound for a given model also holds for the models above and to the left. Similarly, any upper bound for a model also holds for models below and to the right.
For short, we refer to the $k$-bounded combinatorial auction problem as $k$BCA.

\begin{table}[ht!]
\hspace{-.9cm}
\footnotesize \begin{tabular}{@{}c|cc|cc@{}}
\toprule
\multirow{2}{*}{\begin{tabular}[c]{@{}c@{}}\textbf{Single-}\\\textbf{minded}\end{tabular}} & \multicolumn{2}{c|}{Random Order} & \multicolumn{2}{c}{Adversarial Order} \\
 &
  low &
  upp &
  low &
  upp \\ \midrule
iid &
  \multicolumn{2}{c|}{iid} &
  \multicolumn{2}{c}{} \\
(All $k$) &
  0.7451 ${ }^{\text{[r3]}}$ &
  0.7451 ${ }^{\text{[r1,r2]}}$ & 
   &
   \\ \midrule
\begin{tabular}[c]{@{}c@{}}known\\ dist.\end{tabular} &
  \multicolumn{2}{c|}{Prophet Secretary} &
  \multicolumn{2}{c}{Prophet} \\
($k=1$) &
  0.669 ${ }^{\text{[r21]}}$ &
  0.7254 ${ }^{\text{[r22]}}$ & 
  0.5 ${ }^{\text{[r17]}}$  &
  0.5 ${ }^{\text{[r17]}}$\\
($k=2$) &
  0.474 ${ }^{\text{[r4]}}$&
  $0.7254 { }^{\text{[r22]}}$& 
  0.344  ${ }^{\text{[r4]}}$&
$\frac37=0.42857 { }^{\text{[r5]}}$
 \\
($k\geq 3$) &
  $\frac{1}{k+1}$ ${ }^{\text{[r5]}}$ &
  $O\bigl(\frac{\log\log k}{\log k}\bigr)$ ${ }^{\text{[r6]}}$& 
  $\frac{1}{k+1}$ ${ }^{\text{[r5]}}$  &
  $O\bigl(\frac{\log\log k}{\log k}\bigr)$ ${ }^{\text{[r6]}}$ \\\midrule
s.sample &
  \multicolumn{2}{c|}{Single-sample Prophet Secretary} &
  \multicolumn{2}{c}{Single-sample Prophet} \\
 ($k=1$) &
  0.63518 ${ }^{\text{[r7]}}$  &
  $\ln 2=0.6931$ ${ }^{\text{[r9]}}$ & 
  0.5 ${ }^{\text{[r10]}}$&
  0.5 ${ }^{\text{[r17]}}$\\
  ($k=2$) &
  0.25 ${ }^{\text{[r8]}}\to \bm{0.333}$&
   $\ln 2=0.6931$  ${ }^{\text{[r9]}}$  & 
  0.08578 ${ }^{\text{[r15]}}$  &
$\frac{3}{7}=0.42857\ { }^{\text{[r5]}}$\\
  ($k\geq 3$) &
$\frac{1}{ek}$ ${ }^{\text{[r6]}}\to \bm{\frac{1}{k+1}}$ 
  & $O\bigl(\frac{\log\log k}{\log k}\bigr)$  ${ }^{\text{[r6]}}$  &
  $\Omega\bigl(\frac{1}{k^2}\bigr)$ ${ }^{\text{[r12]}}$  &
  $O\bigl(\frac{\log\log k}{\log k}\bigr)$  ${ }^{\text{[r6]}}$ \\ \midrule
\begin{tabular}[c]{@{}c@{}}Subset\\ Sampling\end{tabular} &
  \multicolumn{2}{c|}{Secretary} &
  \multicolumn{2}{c}{Order-Oblivious} \\
 ($k=1$) &
  $\frac{1}{e}=0.36787\ { }^{\text{[r13]}}$  &
  $\frac{1}{e}=0.36787\ { }^{\text{[r13]}}$ & 
  0.25 ${ }^{\text{[r14]}}$&
   0.25 ${ }^{\text{[r14]}}$ \\
  ($k=2$) &
  0.25 ${ }^{\text{[r8]}}$ &
  $\frac{1}{e}=0.36787\ { }^{\text{[r13]}}$ & 
  0.08578  ${ }^{\text{[r15]}}$ &
   0.25 ${ }^{\text{[r14]}}$  \\
  ($k\geq 3$) &
  $\frac{1}{ek}$ ${ }^{\text{[r6]}} \to \bm{k^{-\frac{k}{k-1}}}\ { }^{\text{[r8]}}$ &
  $O\bigl(\frac{\log\log k}{\log k}\bigr)$ ${ }^{\text{[r6]}}$ &
  $\Omega\bigl(\frac{1}{k^2}\bigr)$ ${ }^{\text{[r12]}}$  &
  $O\bigl(\frac{\log\log k}{\log k}\bigr)$  ${ }^{\text{[r6]}}$ \\ \bottomrule
 \toprule
\multirow{2}{*}{\begin{tabular}[c]{@{}c@{}}\textbf{General}\\\textbf{$k$BCA}\end{tabular}} & \multicolumn{2}{c|}{Random Order} & \multicolumn{2}{c}{Adversarial Order} \\
 &
  low &
  upp &
  low &
  upp \\ \midrule
iid &
  \multicolumn{2}{c|}{iid} &
  \multicolumn{2}{c}{} \\
($k=1$) &
   $1-\frac1e=0.6321$ ${ }^{\text{[r16]}}$ &
  0.703 ${ }^{\text{[r18]}}$ &
   &
   \\
 $(k=2)$ &
 $0.333$ ${ }^{\text{[r5]}}\to \bm{0.432=\frac{1-e^{-2}}{2}}$ &
 $0.703$ ${ }^{\text{[r18]}}$ & &   \\
   ($k\ge 3$) &
  $\frac{1}{k+1}$ ${ }^{\text{[r5]}}\to \bm{\frac{1-e^{-k}}{k}}$ & $\bm{O(\frac{\log k}{k})}$ & &\\
   \midrule
\begin{tabular}[c]{@{}c@{}}known\\ dist.\end{tabular} &
  \multicolumn{2}{c|}{Prophet Secretary} &
  \multicolumn{2}{c}{Prophet} \\
  ($k=1$) &
  $1-\frac1e=0.6321$ ${ }^{\text{[r20]}}$ &
  0.703 ${ }^{\text{[r18]}}$ & 
  0.5 ${ }^{\text{[r5]}}$  &
  0.5 ${ }^{\text{[r17]}}$\\
  ($k=2$) &
  0.333 ${ }^{\text{[r5]}}$& $0.703$ ${ }^{\text{[r18]}}$ &
  0.333  ${ }^{\text{[r5]}}$ &
$\frac{3}{7}=0.42857\ { }^{\text{[r5]}}$
 \\
  ($k\geq 3$) &
  $\frac{1}{k+1}$ ${ }^{\text{[r5]}}$ &
  $\bm{O\bigl(\frac{\log k}{k}\bigr)}\gets O\bigl(\frac{\log\log k}{\log k}\bigr)$ ${ }^{\text{[r6]}}$&
  $\frac{1}{k+1}$ ${ }^{\text{[r5]}}$  &
 $\bm{O\bigl(\frac{\log k}{k}\bigr)}\gets O\bigl(\frac{\log\log k}{\log k}\bigr)$ ${ }^{\text{[r6]}}$ \\ \midrule
s.sample &
  \multicolumn{2}{c|}{Single-sample Prophet Secretary} &
  \multicolumn{2}{c}{Single-sample Prophet} \\
 ($k=1$) &
  $\frac{1}{e}=0.36787\  { }^{\text{[r6]}} \to \bm{0.5}$  &
  $\ln 2=0.6931 { }^{\text{[r9]}}$ & 
  0.1715 ${ }^{\text{[r15]}}$&
  0.5 ${ }^{\text{[r17]}}$\\
  ($k=2$) &
  $\frac{1}{2e}=0.183939\ { }^{\text{[r6]}} \to \bm{0.333}$& 
$\ln2 = 0.6931 { }^{\text{[r9]}}$   &
$0.01388$  ${ }^{\text{[r12]}}$  &
$\frac{3}{7}=0.42857\ { }^{\text{[r5]}}$\\
  ($k\geq 3$) &
  $\frac{1}{ek}\ { }^{\text{[r6]}}\to \pmb{\frac{1}{k+1}}$  
  & $\bm{O\bigl(\frac{\log k}{k}\bigr)} \gets O\bigl(\frac{\log\log k}{\log k}\bigr)$ ${ }^{\text{[r6]}}$  & 
  $\Omega\bigl(\frac{1}{k^2}\bigr)$ ${ }^{\text{[r12]}}$  &
 $\bm{O\bigl(\frac{\log k}{k}\bigr)}\gets O\bigl(\frac{\log\log k}{\log k}\bigr)$ ${ }^{\text{[r6]}}$ \\ \midrule
\begin{tabular}[c]{@{}c@{}}Subset\\ Sampling\end{tabular} &
  \multicolumn{2}{c|}{Secretary} &
  \multicolumn{2}{c}{Order-Oblivious} \\
 ($k=1$) &
  $\frac{1}{e}=0.36787\ { }^{\text{[r6]}}$  &
  $\frac{1}{e}=0.36787\ { }^{\text{[r13]}}$ & 
  0.1715 ${ }^{\text{[r15]}}$&
   0.25 ${ }^{\text{[r14]}}$ \\
  ($k=2$) &
$\frac{1}{2e}=0.183939\ { }^{\text{[r6]}}\to \bm{0.25}$ &
  $\frac{1}{e}=0.36787\ { }^{\text{[r13]}}$ & 
  0.01388  ${ }^{\text{[r12]}}$ &
   0.25 ${ }^{\text{[r14]}}$  \\
  ($k\geq 3$) & 
  $\frac{1}{ek}$ ${ }^{\text{[r6]}} \to \bm{k^{-\frac{k}{k-1}}}$ & 
  $\bm{O\bigl(\frac{\log k}{k}\bigr)} \gets O\bigl(\frac{\log\log k}{\log k}\bigr)$ ${ }^{\text{[r6]}}$ & 
  $\Omega\bigl(\frac{1}{k^2}\bigr)$ ${ }^{\text{[r12]}}$  &
 $\bm{O\bigl(\frac{\log k}{k}\bigr)}\gets O\bigl(\frac{\log\log k}{\log k}\bigr)$ ${ }^{\text{[r6]}}$ \\ \bottomrule
    \hline \multicolumn{5}{c}{
\footnotesize\begin{tabular}[c]{@{}l@{}}
\textbf{Table references:}
 r1: Hill and Kertz \cite{Hill1982};\  r2: Kertz \cite{Kertz1986};\ r3: Correa et al. \cite{Correa2017PostedCustomers}; r4 MacRury et al. \cite{MacRury2023};\\
r5: Correa et al. \cite{Correa2022CFPW};
r6: Kesselheim et al. \cite{KesselheimRTV13}; r7: Correa et al. \cite{CorreaCES2020}; r8: Ezra et al. \cite{Ezra2020};\\ r9: Correa et al. \cite{Correa2019ProphetDistribution}; 
r10: Rubinstein et al.\cite{Rubinstein2020OptimalSamples}; r11: Correa et al.\cite{Correa2022CFPW};\\r12: Korula and Pál\cite{KP2019};  r13: Dynkin \cite{Dynkin1963}; r14: folklore; r15: Kaplan et al. \cite{Kaplan2022OnlineSample};\\
r16: Brubach et al. \cite{Brubach2016NewMatching};  r17: Krengel and Sucheston \cite{Krengel1987ProphetProcesses};  r18: Huang et al. \cite{HuangSY2022};\\
r19: Ezra et al. \cite{Ezra2020OnlineModels}; r20: Ehsani et al. \cite{Ehsani2018}; r21: Correa et al. \cite{CorreaSZ2021}; r22: Bubna and Chiplunkar \cite{bubna2023prophet}
\end{tabular} 
}\\ \bottomrule
\bottomrule
\end{tabular}
\caption{Current and new results for all models; our new  results are in \textbf{bold} font. }
\label{big-table}
\end{table}

\bibliographystyle{acm}
{\footnotesize	\bibliography{references}}

\appendix

\section{Missing Proofs from Section \ref{sec:algorithms}}\label{app:algorithms}

\begin{proof}[Proof of Claim \ref{claim:key-inequality-sample}]
For every $a\in A$, let $\beta_a(\bm{r})$ be the value $r_{a}(e)\cdot\PP\left[\calP_{a}(\bm{r})=(S,e)\right]$.
Observe that
\begin{align}
&m!\;\EE_{\mu}\left[r_{\mu(t)}(e)\cdot\PP\left[\calP_{\mu(t)}(\bm{r})=(S,e)\right]\cdot \prod_{i=j}^{t-1} \phi_i(\mu)\right]\notag\\
&=\sum_{a\in A}\beta_a(\bm{r})\sum_{\substack{\mu \in \bijec([m],A):\\ \mu(t)=a}}\PP[\calE_{\mu_j}(\bm{r}(\mu_j\cup A\setminus \mu_t)\cup \bm{s}(\mu_t\setminus \mu_j))]\cdot \prod_{i=j+1}^{t-1} \phi_i(\mu)\notag\\
&=\sum_{a\in A}\beta_a(\bm{r})\sum_{\substack{B\subseteq A:\\a\in B,\\|B|=t-j}}\sum_{\substack{\mu \in \bijec([m],A):\\ \mu(t)=a,\\ \mu_t\setminus \mu_j=B}}\PP[\calE_{\mu_j}(\bm{r}(A\setminus B)\cup \bm{s}(B))]\cdot \prod_{i=j+1}^{t-1} \phi_i(\mu)\notag\\
&=\sum_{a\in A}\beta_a(\bm{r})\sum_{\substack{B\subseteq A:\\a\in B,\\|B|=t-j}}\sum_{b\in A\setminus B}\PP[\calE_{b}(\bm{r}(A\setminus B)\cup \bm{s}(B))]\hspace{-.3cm}\sum_{\substack{\mu \in \bijec([m],A):\\ \mu(t)=a,\mu(j)=b,\\ \mu_t\setminus \mu_j=B}} \prod_{i=j+1}^{t-1} \phi_i(\mu),\label{sample:partitioning1}
\end{align}
where the first equality holds by partitioning the orders according to $\mu(t)=a\in A$, and expanding the definition of $\phi_j(\mu)$,
the second equality holds by partitioning the orders according to $\mu_t\setminus \mu_j=B$ (i.e., the set of agents presented between $j+1$ and $t$ is exactly the set $B$), and in the last equality, we partition the orders according to the value $\mu(j)=b\in A\setminus B$.

Let us prove that the inner sum above is independent of the agent $b\in A\setminus B$ considered. To see that, let $\mu\in \bijec([m],A)$ be any permutation with $\mu(t)=a$, $\mu(j)=b$ and $\mu_t\setminus \mu_j=B$. And let $c\in A\setminus B$ be any agent (not necessarily distinct from $b$). Let $j'\in \{1,\dots, j\}\cup \{t+1,\dots, m\}$ be the index such that $\mu(j')=c$. Consider the permutation $\nu$ that exchanges the preimages of $b$ and $c$, i.e. $\nu(\ell)=\mu(\ell)$ for all $\ell\not\in \{j,j'\}$, $\nu(j)=\mu(j')=c$, $\nu(j')=\mu(j)=b$. Observe that for all $i\in \{j+1,\dots, t-1\}$ $\nu(i)=\mu(i)$, $\nu_t\setminus \nu_i=\mu_t\setminus \mu_i$ and $\nu_i \cup A\setminus \nu_t = \mu_i \cup A\setminus \mu_t$, and thus, $\phi_i(\nu)=\phi_i(\mu)$. We deduce that
\vspace{-.5cm}

\begin{equation}
\sum_{\substack{\mu \in \bijec([m],A):\\ \mu(t)=a,\\ \mu_t\setminus \mu_j=B}} \prod_{i=j+1}^{t-1} \phi_i(\mu)=(m-t+j)\sum_{\substack{\mu \in \bijec([m],A):\\ \mu(t)=a,\mu(j)=b,\\ \mu_t\setminus \mu_j=B}} \prod_{i=j+1}^{t-1} \phi_i(\mu),\label{sample:mixing-tau}
\end{equation}
where the factor $m-t+j$ comes from the fact that $|A\setminus B|=m-t+j$.
Therefore, using \eqref{sample:mixing-tau}, the expression in \eqref{sample:partitioning1} is equal to
\begin{align}
&\sum_{a\in A}\beta_a(\bm{r})\sum_{\substack{B\subseteq A:\\a\in B,\\|B|=t-j}}\sum_{b\in A\setminus B}\frac{\PP[\calE_{b}(\bm{r}(A\setminus B)\cup \bm{s}(B))]}{m-t+j}\hspace{-.3cm}\sum_{\substack{\mu \in \bijec([m],A):\\ \mu(t)=a,\\ \mu_t\setminus \mu_j=B}} \prod_{i=j+1}^{t-1} \phi_i(\mu)\notag\\
&=\sum_{a\in A}\beta_a(\bm{r})\sum_{\substack{B\subseteq A:\\a\in B,\\|B|=t-j}}\sum_{\substack{\mu \in \bijec([m],A):\\ \mu(t)=a,\\ \mu_t\setminus \mu_j=B}} \prod_{i=j+1}^{t-1} \phi_i(\mu)\sum_{b\in A\setminus B}\frac{\PP[\calE_{b}(\bm{r}(A\setminus B)\cup \bm{s}(B))]}{m-t+j}\notag\\
&\ge \left(1-\frac{k}{m-t+j}\right)_{\!\!+}\sum_{a\in A}\beta_a(\bm{r})\sum_{\substack{B\subseteq A:\\a\in B,\\|B|=t-j}}\sum_{\substack{\mu \in \bijec([m],A):\\ \mu(t)=a,\\ \mu_t\setminus \mu_j=B}} \prod_{i=j+1}^{t-1} \phi_i(\mu)\notag\\
&= \left(1-\frac{k}{m-t+j}\right)_{\!\!+}m!\;\EE_{\mu}\left[r_{\mu(t)}(e)\cdot\PP\left[\calP_{\mu(t)}(\bm{r})=(S,e)\right]\cdot \prod_{i=j+1}^{t-1} \phi_i(\mu)\right],\label{sample:partitioning4}
\end{align}
where the first equality holds since by \eqref{sample:mixing-tau} the right-most summation does not depend on $b$, and the inequality holds since by the blocking property \ref{sampler-c} for the weight profile $\bm{r}(A\setminus B)\cup \bm{s}(B)$ we have
\begin{align*}
\sum_{b\in A\setminus B}\PP[\calE_{b}(\bm{r}(A\setminus B)\cup \bm{s}(B))]&=\sum_{b\in A\setminus B}(1-\PP[\calP_b(\bm{r}(A\setminus B)\cup \bm{s}(B))\text{ blocks }(S,e)])\\
&=m-t+j-\sum_{b\in A\setminus B}\PP[\calP_b(\bm{r}(A\setminus B)\cup \bm{s}(B))\text{ blocks }(S,e)]\\
&\ge (m-t+j-k)_+.
\end{align*}
The second equality in \eqref{sample:partitioning4} holds by undoing the partitioning according to $\mu_t\setminus \mu_j$.
This finishes the proof of the claim.
\end{proof}

\begin{proof}[Proof of Claim \ref{claim:falling}]
First observe that for $t\in [m]$,
\begin{align}
\prod_{i=1}^{t-1}\left(1-\frac{k}{m-t+i}\right)_+=\prod_{\ell=m-t+1}^{m-1}\left(\frac{\ell -k}{\ell}\right)_+= \frac{\binom{m-k-1}{t-1}}{\binom{m-1}{t-1}} =\begin{cases}
1 &\quad \text{if $t=1$}\\
0 &\quad \text{if $t\geq 2, m\leq k$}\\
\frac{\binom{m-t}{k}}{\binom{m-1}{k}} &\quad \text{if $t\geq 2, m\geq k+1$}
\end{cases}
\label{eq:claim2}
\end{align}
where we use the standard convention that $\binom{a}{0}=1$ for all $a\in \ZZ$ and $\binom{a}{b}=0$ for all $a<0$ and $b\geq 0$.

Then, if $m\geq k+1$, using the Hockey-stick identity for binomial coefficients,
\begin{align*}
\frac{1}{m}\sum_{t=1}^m\; \prod_{i=1}^{t-1}\left(1-\frac{k}{m-t+i}\right)_+&=\frac{1}{m\binom{m-1}{k}}\sum_{t=1}^m \binom{m-t}{k} = \frac{1}{m\binom{m-1}{k}} \binom{m}{k+1}=\frac{1}{k+1}.
\end{align*}
And if $m\leq k$, $\frac{1}{m}\sum_{t=1}^m\; \prod_{i=1}^{t-1}\left(1-\frac{k}{m-t+i}\right)_+=\frac{1}{m}\geq \frac{1}{k+1}.$
\end{proof}

\begin{proof}[Proof of Claim \ref{claim:key-inequality}]
Observe that for every $j\in \{\tau+1,\ldots,t-1\}$, we have
\begin{align}
\sum_{\mu \in U(a,t,T)}\prod_{i=j}^{t-1}\PP[\calC_i(\mu)]
&=\sum_{\substack{T'\subseteq T:\\|T'|=j}}\;\sum_{b\in T'}\sum_{\substack{\mu \in U(a,t,T):\\\ \mu_{j}=T'\\ \mu(j)=b}}\PP\left[\calP_{\mu(j)}(\bm{w}(T')) \text{ does not block } (S,e)\right]\cdot \prod_{i=j+1}^{t-1}\PP[\calC_i(\mu)]\notag\\
&=\sum_{\substack{T'\subseteq T:\\|T'|=j}}\sum_{b\in T'}\PP\left[\calP_{b}(\bm{w}(T')) \text{ does not block } (S,e)\right] \sum_{\substack{\mu \in U(a,t,T):\\\ \mu_{j}=T'\\ \mu(j)=b}}\prod_{i=j+1}^{t-1}\PP[\calC_i(\mu)],\label{sec:LB1}
\end{align}
where the first equality holds by partitioning the space according to the set $T'\subseteq T$ such that $\mu_{j}=T'$ and the value $b\in T'$ that can be taken by $\mu(j)$; and the last holds by rearranging the summation terms.
We claim that the value of the inner summation is independent of the choice of $b\in T'$. To see this, let $\mu \in U(a,t,T)$  with $\mu_j=T'$ and $\mu(j)=b$. Let $c\in T'$ be any agent (not necessarily distinct from $b$) and let $j'\in \{1,\dots, t\}$ be the index such that $\mu(j')=c$ 
Consider the permutation $\nu\in U(a,t,T)$ with $\nu_j=T'$ that exchanges the preimages of $b$ and $c$, that is $\nu(\ell)=\mu(\ell)$ for all $\ell\not\in \{j,j'\}$, $\nu(j)=\mu(j')=c$, $\nu(j')=\mu(j)=b$. Observe that for each $i\in \{j+1,\dots, t-1\}$, $\nu_i=\mu_i$ and $\nu(i)=\mu(i)$, and thus, the events $\calC_i(\mu)$ and $\calC_i(\nu)$ are the same. We deduce that 

\begin{equation}
\sum_{\substack{\mu \in U(a,t,T):\\\ \mu_{j}=T'}} 
\prod_{i=j+1}^{t-1}\PP[\calC_i(\mu)]= j\cdot \sum_{\substack{\mu \in U(a,t,T):\\\ \mu_{j}=T'\\ \mu(j)=b}}\prod_{i=j+1}^{t-1}\PP[\calC_i(\mu)].\label{sec:mixing-tau}
\end{equation}
Then, using \eqref{sec:mixing-tau}  we can lower-bound \eqref{sec:LB1} by
\begin{align}
&\sum_{\substack{T'\subseteq T:\\|T'|=j}}\sum_{b\in T'}\PP\left[\calP_{b}(\bm{w}(T')) \text{ does not block } (S,e)\right]\cdot \frac{1}{j}\sum_{\substack{\mu \in U(a,t,T):\\\ \mu_{j}=T'}}\prod_{i=j+1}^{t-1}\PP[\calC_i(\mu)]\notag\\
&=\sum_{\substack{T'\subseteq T:\\|T'|=j}}\sum_{\substack{\mu \in U(a,t,T):\\\ \mu_{j}=T'}}\left(\prod_{i=j+1}^{t-1}\PP[\calC_i(\mu)]\right)\cdot \Big(1-\frac{1}{j}\sum_{b\in T'}\PP\left[\calP_{b}(\bm{w}(T')) \text{ blocks } (S,e)\right]\Big)\notag\\
&\ge \left(1-\frac{k}{j}\right)_{\!+}\;\sum_{\substack{T'\subseteq T:\\|T'|=j}}\sum_{\substack{\mu \in U(a,t,T):\\\ \mu_{j}=T'}}\prod_{i=j+1}^{t-1}\PP[\calC_i(\mu)]= \left(1-\frac{k}{j}\right)_{\!+}\;\sum_{\mu \in U(a,t,T)}\prod_{i=j+1}^{t-1}\PP[\calC_i(\mu)],
\end{align}
where the first equality follows since $|T'|=j$ and by taking the probability of the complementary events for each $b\in T'$; the inequality holds by using the blocking property \ref{sampler-b} of the certificate sampler when applied with the weight profile $\bm{w}(T')$ for the set of agents $T'$, and the second equality follows by undoing the space partitioning.
This finishes the proof of the claim.
\end{proof}

\end{document}